\documentclass[a4paper,UKenglish]{lipics}
 
\usepackage{microtype}
\usepackage{multirow}
\usepackage{color}
\usepackage[noend]{algorithmic}
\usepackage{algorithm}
\usepackage{graphicx}
\usepackage{graphics}
\usepackage{pgfplots}
\usepackage{refcount}
\pgfplotsset{filter discard warning=false}
\usepackage{tikz}
\usetikzlibrary{shapes,decorations,arrows,snakes,calc,positioning}
\usepackage{cite}
\usepackage{etoolbox}
\usepackage{thmtools}
\usepackage{thm-restate}
\usepackage{wrapfig}

\newcommand{\bid}[1]{\mbox{\textbf{#1}}}
\newcommand{\anset}[1]{{{\langle #1 \rangle}}}

\newcommand{\concat}{\circ}

\def\popq{\bid{remove}}

\def\pushq{\bid{insert}}

\def\nexti{\bid{next}}

\def\pkwd{{\mathit{keyword}}}
\def\itgrs{{\mathit{pathGroups}}}

\def\parents{\bid{parents}}
\def\prodAns{\bid{produceAnswers}}

\def\continue{\bid{continue}}
\def\true{\bid{true}}
\def\false{\bid{false}}
\def\firstp{{\bid{first}}}
\def\predp{{\bid{predecessor}}}

\def\fweight{w}

\def\act{\textit{active}}
\def\vis{\textit{visited}}
\def\inan{\textit{in-answer}}

\def\V{V}

\def\G{G}

\title{A Practically Efficient Algorithm for Generating Answers to Keyword Search over Data
Graphs\footnote{This work was supported by the Israel Science Foundation (Grant No.~1632/12).}\footnote{This paper is the full version of~\cite{ICDT16}.}}
\titlerunning{An Algorithm for Keyword Search over Data Graphs} 

\author{Konstantin Golenberg}
\author{Yehoshua Sagiv}
\affil{The Hebrew University of Jerusalem, Israel}

\Copyright{Konstantin Golenberg and Yehoshua Sagiv}


\begin{document}

\maketitle
\begin{abstract} 
  In keyword search over a data graph, an answer is a non-redundant
  subtree that contains all the keywords of the query. A naive
  approach to producing all the answers by increasing height is to
  generalize Dijkstra's algorithm to enumerating all acyclic paths by
  increasing weight. The idea of \emph{freezing} is introduced so that
  (most) non-shortest paths are generated only if they are actually
  needed for producing answers. The resulting algorithm for generating
  subtrees, called \emph{GTF}, is subtle and its proof of correctness
  is intricate.  Extensive experiments show that GTF outperforms
  existing systems, even ones that for efficiency's sake are
  incomplete (i.e.,~cannot produce all the answers).  In particular,
  GTF is scalable and performs well even on large data graphs and when
  many answers are needed.
\end{abstract}

\section{\label{sec:intro}Introduction}

Keyword search over data graphs is a convenient paradigm of querying
semistructured and linked data. Answers, however, are similar to those
obtained from a database system, in the sense that they are succinct
(rather than just relevant documents) and include semantics (in the
form of entities and relationships) and not merely free text. Data
graphs can be built from a variety of formats, such as XML, relational
databases, RDF and social networks.  They can also be obtained from
the amalgamation of many heterogeneous sources.  When it comes to
querying data graphs, keyword search alleviates their lack of
coherence and facilitates easy search for precise answers, as if
users deal with a traditional database system.

In this paper, we address the issue of efficiency. Computing keyword
queries over data graphs is much more involved than evaluation of
relational expressions.  Quite a few systems have been developed
(see~\cite{tkdeCW14} for details). However, they 
fall short of the degree of efficiency and scalability that is
required in practice.  Some algorithms sacrifice \emph{completeness}
for the sake of efficiency; that is, they are not capable of
generating all the answers and, consequently, may miss some relevant ones.

We present a novel algorithm,
called \emph{Generating Trees with Freezing} (GTF).
We start with a straightforward generalization of Dijkstra's
shortest-path algorithm to the task of constructing all simple
(i.e.,~acyclic) paths, rather than just the shortest ones.  Our main
contribution is incorporating the \emph{freezing} technique that
enhances efficiency by up to one \linebreak\linebreak\linebreak order of magnitude, compared with the
naive generalization of Dijkstra's algorithm. The main idea is to
avoid the construction of most non-shortest paths until they are
actually needed in answers. Freezing may seem intuitively clear, but
making it work involves subtle details and requires an intricate proof
of correctness.

Our main theoretical contribution is the algorithm GTF, which
incorporates freezing, and its proof of correctness.  Our main
practical contribution is showing experimentally (in
Section~\ref{sec:experiments_sum} 
and Appendix~\ref{sec:experiments}
)
that GTF is both more efficient and more scalable than existing
systems. This contribution is especially significant in light of the
following.  First, GTF is complete (i.e.,~it does not miss answers);
moreover, we show experimentally that not missing answers is important
in practice.  Second, the order of generating answers is by increasing
height.  This order is commonly deemed a good strategy for an initial
ranking that is likely to be in a good correlation with the final one
(i.e.,~by increasing weight).

\tikzstyle{vertex}=[ellipse,draw,minimum size=14pt,inner sep=1pt]   
\tikzstyle{keyword}=[rectangle,draw,minimum size=1.5pt,inner sep=2pt]  
\usetikzlibrary{positioning}

\begin{figure*}
\begin{minipage}{0.59\textwidth}
\begin{center}
  \begin{tikzpicture}[->,>=stealth',shorten >=0pt,auto,node distance=0.3cm]

  
  \node[vertex] (city) at (0,0) {$\mathit{city}$};
  \node[vertex] (prov) [below left = of city] {$\mathit{province}$};
  \node[vertex] (country) [below right = of city] {$\mathit{country}$};
  \node[vertex] (river) [dashed, left = 1 of city] {$\mathit{river}$};
  \node (dummy) at (0,0.7) {};
  
  \node[keyword] (paris) [right = of city] {$\mathit{Paris}$};  
  \node[keyword] (fr) [below  = 0.75 of city] {$\mathit{France}$};  
  \node[keyword] (ile) [below left = of prov] {$\mathit{Ile}$};  
  \node[keyword] (de) [below = 0.15 of prov] {$\mathit{de}$};  
  \node[keyword] (seine) [dashed, left = 0.7 of river] {$\mathit{Seine}$};
    
  \draw (city) edge (prov);
  \draw (city) edge (country);
  \draw (prov) edge (country);

  \draw (city) edge (paris);
  \draw (prov) edge (ile);
  \draw (prov) edge (de);
  \draw (prov) edge (fr);
  \draw (country) edge (fr);
  
  \draw (river) [dashed] edge (city);
  \draw (river) [dashed] edge (seine);

	\end{tikzpicture}
	\vskip1em
   \caption{A snippet of a data graph}\label{fig:datagraph}
\end{center}
\end{minipage}
\begin{minipage}{0.40\textwidth}
\begin{center}
  \begin{tikzpicture}[->,>=stealth',shorten >=0pt,auto,node distance=0.3cm]
  

  \node[vertex] (city_a4) at (0,0) {$\mathit{city}$};
  \node[vertex] (prov_a4) [below right = of city_a4] {$\mathit{province}$};
  \node[vertex] (river) [above = of city_a4] {$\mathit{river}$};
    
  \node[keyword] (paris_a4) [below left = of city_a4] {$\mathit{Paris}$};  
  \node[keyword] (fr_a4) [below  = of  prov_a4] {$\mathit{France}$};    
  
  \draw (city_a4) edge (prov_a4);
  \draw (city_a4) edge (paris_a4);
  \draw (prov_a4) edge (fr_a4);
  \draw (river) edge (city_a4);
  
  \end{tikzpicture}
  \caption{Redundant subtree\label{fig:redun_ans}}
\end{center}
\end{minipage}
\end{figure*}

\begin{figure}
  \centering
  \begin{tikzpicture}[->,>=stealth',shorten >=0pt,auto,node distance=0.3cm]
  

  \node[vertex] (city_a1) at (0,0) {$\mathit{city}$};
  \node[vertex] (prov_a1) [below right = of city_a1] {$\mathit{province}$};
    
  \node[keyword] (paris_a1) [below left = of city_a1] {$\mathit{Paris}$};  
  \node[keyword] (fr_a1) [below  = of  prov_a1] {$\mathit{France}$};    
  \node (a1) [below = 2 of city_a1] {$A_1$};
  
  \draw (city_a1) edge (prov_a1);
  \draw (city_a1) edge (paris_a1);
  \draw (prov_a1) edge (fr_a1);


  \node[vertex] (city_a2) at (4.5,0) {$\mathit{city}$};
  \node[vertex] (country_a2) [below right = of city_a2] {$\mathit{country}$};
    
  \node[keyword] (paris_a2) [below left = of city_a2] {$\mathit{Paris}$};  
  \node[keyword] (fr_a2) [below  = of  country_a2] {$\mathit{France}$};  
    
  \draw (city_a2) edge (country_a2);
  \draw (city_a2) edge (paris_a2);
  \draw (country_a2) edge (fr_a2);
  \node (a2) [below = 2 of city_a2] {$A_2$};
  
  \node[vertex] (city_a3) at (9,0) {$\mathit{city}$};
  \node[vertex] (prov_a3) [below right = of city_a3] {$\mathit{province}$};
  \node[vertex] (country_a3) [below = of prov_a3] {$\mathit{country}$};
    
  \node[keyword] (paris_a3) [below left = of city_a3] {$\mathit{Paris}$};  
  \node[keyword] (fr_a3) [below  = of  country_a3] {$\mathit{France}$};    
  \node (a3) [below = 2 of city_a3] {$A_3$};
  
  \draw (city_a3) edge (prov_a3);
  \draw (city_a3) edge (paris_a3);
  \draw (prov_a3) edge (country_a3);
  \draw (country_a3) edge (fr_a3);
  
  \end{tikzpicture}   
  \caption{Answers}\label{fig:answers}
  \vspace{-1em}
\end{figure}

\section{Preliminaries}\label{sec:prelim}
We model data as a directed graph $G$, similarly to~\cite{icdeBHNCS02}.  
Data graphs can be constructed from a variety of formants (e.g.,~RDB, XML 
and RDF). 
Nodes represent entities and relationships, while edges correspond to
connections among them (e.g.,~foreign-key references when the data
graph is constructed from a relational database).  We assume that text
appears only in the nodes. This is not a limitation, because we can
always split an edge (with text) so that it passes through a
node. Some nodes are for keywords, rather than entities and
relationships. In particular, for each keyword $k$ that appears in the
data graph, there is a dedicated node.  By a slight abuse of notation,
we do not distinguish between a keyword $k$ and its node---both are
called \emph{keyword} and denoted by $k$. For all nodes $v$ of the
data graph that contain a keyword $k$, there is a directed edge from
$v$ to $k$.  Thus, keywords have only incoming edges.

Figure~\ref{fig:datagraph} shows a snippet of a data graph.
The dashed part should be ignored unless explicitly stated otherwise.
Ordinary nodes are shown as
ovals.  For clarity, the type of each node appears inside the oval.
Keyword nodes are depicted as rectangles. To keep the figure small,
only a few of the keywords that appear in the graph are shown as
nodes. For example, a type is also a keyword and has its own node in the full
graph. For each oval, there is an edge to every keyword that it contains.

Let $G=(V,E)$ be a directed data graph, where $V$ and $E$ are the sets
of nodes and edges, respectively.
A directed path is denoted by $\langle v_1,\ldots,v_m \rangle$.
We only consider \emph{rooted} (and, hence, directed) subtrees $T$ of
$\G$. That is, $T$ has a unique node $r$, such that for all nodes $u$
of $T$, there is exactly one path in $T$ from $r$ to $u$.
 Consider a query $K$, that is, a set of at least two
keywords.  A \emph{$K$-subtree} is a rooted subtree of $G$, such that
its leaves are exactly the keywords of $K$.  We say that a node $v\in
V$ is a \emph{$K$-root} if it is the root of some $K$-subtree of $G$.  It is
observed in~\cite{icdeBHNCS02} that $v$ is a $K$-root if and only if
for all $k\in K$, there is a path in $G$ from $v$ to $k$.
An \emph{answer} to $K$ is a $K$-subtree $T$ that is \emph{non-redundant}
(or \emph{reduced}) in the sense that no proper subtree $T'$ of $T$ is
also a $K$-subtree. It is easy to show that a $K$-subtree $T$ of
$G$ is an answer 
if and only if the root of $T$ has at least two
children. Even if $v$ is a $K$-root, it does not necessarily follow
that there is an answer to $K$ that is rooted at $v$ (because it is possible
that in all $K$-subtrees rooted at $v$, there is only one child of $v$).

Figure~\ref{fig:answers} shows three answers to the query
$\left\{\mathit{France},\mathit{Paris}\right\}$ over the data graph of
Figure~\ref{fig:datagraph}.  The answer $A_1$ means that the city Paris is
located in a province containing the word France in its name.
The answer $A_2$ states that the city Paris is located in
the country France.  Finally, the answer $A_3$ means that Paris is
located in a province which is located in France.  

Now, consider also the dashed part of Figure~\ref{fig:datagraph}, that
is, the keyword $\mathit{Seine}$ and the node $\mathit{river}$ with
its outgoing edges.  There is a path from $\mathit{river}$ to every keyword
of $K=\left\{\mathit{France},\mathit{Paris}\right\}$. Hence,
$\mathit{river}$ is a $K$-root. However, the $K$-subtree of
Figure~\ref{fig:redun_ans} is not an answer to $K$, 
because its root has only one child.

For ranking, the nodes and edges of the data graph have positive
weights.  The \emph{weight} of a path (or a tree) is the sum of
weights of all its nodes and edges. The rank of an answer is inversely
proportional to its weight.  The \emph{height} of a tree is the
maximal weight over all paths from the root to any leaf
(which is a keyword of the query).
For example, suppose that the weight of each node and edge is $1$.
The heights of the answers $A_1$ and $A_3$ (of
Figure~\ref{fig:answers}) are $5$ and $7$, respectively. In $A_1$, the
path from the root to $\mathit{France}$ is a minimal
(i.e.,~shortest) one between these two nodes, in the whole graph, and
its weight is $5$.  In $A_3$, however, the path from the root (which
is the same as in $A_1$) to $\mathit{France}$ has a higher weight, namely,~$7$.

\section{The GTF Algorithm}\label{sec:gtf}
\subsection{The Naive Approach}\label{sec:naive}

Consider a query $K=\left\{{k_1,\ldots,k_n}\right\}$.
In~\cite{icdeBHNCS02}, they use a backward shortest-path iterator from
each keyword node $k_i$. That is, starting at each $k_i$, they apply
Dijkstra's shortest-path algorithm in the opposite direction of the
edges. If a node $v$ is reached by the backward iterators from all the $k_i$,
then $v$ is a $K$-root (and, hence, might be the root of some answers).
In this way, answers are generated by increasing height.
However, this approach can only find answers that consist of 
shortest paths from the root to the keyword nodes.
Hence, it misses answers (e.g.,~it cannot produce $A_3$ of 
Figure~\ref{fig:answers}).

Dijkstra's algorithm can be straightforwardly generalized to construct
all the simple (i.e.,~acyclic) paths by increasing weight.  
This approach is used\footnote{They used it on a small \emph{summary} graph
to construct database queries from keywords.}
in~\cite{icdeTWRC09} and it consists of two parts: 
path construction and answer production.  Each constructed path is from
some node of $G$ to a keyword of $K$. Since paths are constructed
backwards, the algorithm starts simultaneously from all the keyword
nodes of $K$. It uses a single priority queue to generate, by increasing weight,
all simple paths to every keyword node of $K$.
When the algorithm discovers that a node $v$ is a
$K$-root (i.e.,~there is a path from $v$ to every $k_i$), it
starts producing answers rooted at $v$. This is done by considering
every combination of paths $p_1,\ldots,p_n$, such that $p_i$ is from
$v$ to $k_i$ ($1\le i\le n$). If the combination is a non-redundant
$K$-subtree of $G$, then it is produced as an answer.
It should be noted that in~\cite{icdeTWRC09}, answers are subgraphs;
hence, every combination of paths $p_1,\ldots,p_n$ is an answer.
We choose to produce subtrees as answers for two reasons.
First, in the experiments of Section~\ref{sec:experiments_sum}, we compare
our approach with other systems that produce subtrees.
Second, it is easier for users to understand answers that are
presented as subtrees, rather than subgraphs.

\tikzstyle{vertex}=[ellipse,draw,minimum size=14pt,inner sep=1pt]  
\tikzstyle{keyword}=[rectangle,draw,minimum size=1.5pt,inner sep=2pt]  
\usetikzlibrary{positioning}

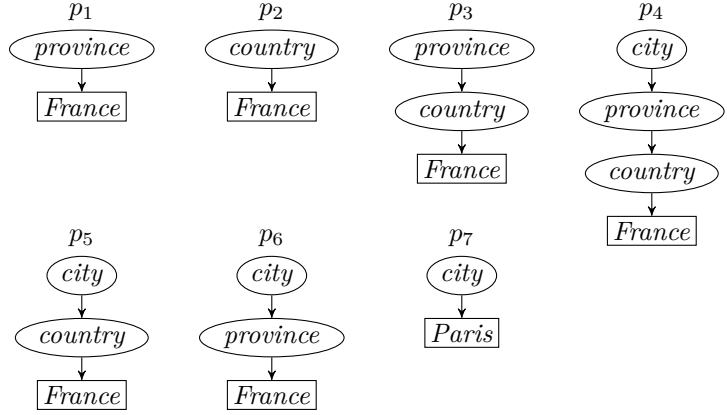
\begin{figure}[tb]
\centering
\begin{tikzpicture}[->,>=stealth',shorten >=0pt,auto,node distance=0.3cm]

  
  \node[vertex] (prov) at (0,0) {$\mathit{province}$};
  \node[keyword] (fr) [below  = of prov] {$\mathit{France}$};  
  \draw (prov) edge (fr);

  \node[vertex] (coun) at (2.5,0) {$\mathit{country}$};
  \node[keyword] (fr) [below  = of coun] {$\mathit{France}$};  
  \draw (coun) edge (fr);

  \node[vertex] (prov) at (5,0) {$\mathit{province}$};
  \node[vertex] (coun) [below  = of prov] {$\mathit{country}$};
  \node[keyword] (fr) [below  = of coun] {$\mathit{France}$};  
  \draw (prov) edge (coun);
  \draw (coun) edge (fr);
  
  \node[vertex] (city) at (7.5,0) {$\mathit{city}$};
  \node[vertex] (prov) [below  = of city] {$\mathit{province}$};
  \node[vertex] (coun)  [below  = of prov] {$\mathit{country}$};
  \node[keyword] (fr) [below  = of coun] {$\mathit{France}$};  
  \draw (city) edge (prov);
  \draw (prov) edge (coun);
  \draw (coun) edge (fr);
  

  \node[vertex] (city) at (0,-3) {$\mathit{city}$};
  \node[vertex] (coun) [below  = of city] {$\mathit{country}$};
  \node[keyword] (fr) [below  = of coun] {$\mathit{France}$};  
  \draw (city) edge (coun);
  \draw (coun) edge (fr);

  \node[vertex] (city) at (2.5,-3) {$\mathit{city}$};
  \node[vertex] (prov) [below  = of city] {$\mathit{province}$};
  \node[keyword] (fr) [below  = of prov] {$\mathit{France}$};  
  \draw (city) edge (prov);
  \draw (prov) edge (fr);

  \node[vertex] (city) at (5,-3) {$\mathit{city}$};
  \node[keyword] (par) [below  = of city] {$\mathit{Paris}$};  
  \draw (city) edge (par);

  \node (p1) at (0,0.5) {$p_1$};
  \node (p2) at (2.5,0.5) {$p_2$};
  \node (p3) at (5,0.5) {$p_3$};
  \node (p4) at (7.5,0.5) {$p_4$};
  
  \node (p5) at (0,-2.5) {$p_5$};
  \node (p6) at (2.5,-2.5) {$p_6$};
  \node (p7) at (5,-2.5) {$p_7$};
\end{tikzpicture}
\caption{\label{fig:gt} Paths to keywords in the graph snippet of Figure~\ref{fig:datagraph}}
\end{figure}

The drawback of the above approach is constructing a large number of
paths that are never used in any of the generated answers.  To
overcome this problem, the next section introduces the technique of
\emph{freezing}, thereby most non-minimal paths are generated only if
they are actually needed to produce answers.
Section~\ref{sec:pseudocode} describes the algorithm
\emph{Generating Trees with Freezing} (GTF)
that employs this technique.

To save space (when constructing all simple paths),
we use the common technique known as \emph{tree of paths}.  
In particular, a path $p$ is a linked list,
such that its first node points to the rest of $p$.
As an example, consider the graph snippet of Figure~\ref{fig:datagraph}.
The paths that lead to the keyword $\mathit{France}$ are 
$p_1$,~$p_2$, $p_3$, $p_4$, $p_5$~and $p_6$, shown in Figure~\ref{fig:gt}.
Their tree of paths is presented in Figure~\ref{fig:pathsTree}.

Since we build paths backwards, a data graph is preprocessed to
produce for each node $v$ the set of its \emph{parents}, that is, the
set of nodes $v'$, such that $(v',v)$ is an edge of the data graph. We
use the following notation.  Given a path $p$ that starts at a node
$v$, the extension of $p$ with a parent $v'$ of $v$ is denoted by $v'
\rightarrow p$. Note that $v'$ is the first node of $v' \rightarrow p$
and $v$ is the second one.

\tikzstyle{vertex}=[ellipse,draw,minimum size=14pt,inner sep=1pt]  
\tikzstyle{keyword}=[rectangle,draw,minimum size=1.5pt,inner sep=2pt]  
\usetikzlibrary{positioning}

\begin{figure}
\centering
\begin{tikzpicture}[->,>=stealth',shorten >=0pt,auto,node distance=0.3cm]
  
  \node[keyword] (k1) at (0,0) {$\mathit{France}$};
  \node[vertex] (v2p'1) [above right = of k1] {$\mathit{province}$ };
  \node[right = 0 of v2p'1]{$p_1$};

  \node[vertex] (v1p1) [above right = of v2p'1] {$\mathit{city}$ };
  \node[right = 0 of v1p1]{$p_6$};

  \node[vertex] (v3p'3) [above left = of k1] {$\mathit{country}$ };
  \node[left = 0 of v3p'3]{$p_2$};

  \node[vertex] (v1p3) [above right = of v3p'3] {$\mathit{city}$ };
  \node[right = 0 of v1p3]{$p_5$};

  \node[vertex] (v2p'2) [above left = of v3p'3] {$\mathit{province}$ };
  \node[left = 0 of v2p'2]{$p_3$};

  \node[vertex] (v1p2) [above = of v2p'2] {$\mathit{city}$ };
  \node[left = 0 of v1p2]{$p_4$};

  \draw(v2p'1) edge (k1);
  \draw(v3p'3) edge (k1);
  \draw(v1p1)  edge (v2p'1);
  \draw(v1p3)  edge (v3p'3);
  \draw(v2p'2) edge (v3p'3);
  \draw(v1p2)  edge (v2p'2);

\end{tikzpicture}
\caption{\label{fig:pathsTree} Tree of paths}
\end{figure}

\subsection{Incorporating Freezing}

The general idea of freezing is to avoid the construction of paths that cannot contribute to production of answers.
To achieve that, a non-minimal path $p$ is frozen
until it is certain that $p$ can reach (when constructed backwards)
a $K$-root.
In particular, the first path that reaches a node $v$ is always a minimal one.
When additional paths reach $v$, they are frozen there until
$v$ is discovered to be on a path from a $K$-root to a keyword node.
The process of answer production in the GTF algorithm remains the same as in the naive approach.

We now describe some details about the implementation of GTF.  We mark
nodes of the data graph as either $\act$, $\vis$ or $\inan$.  Since we
simultaneously construct paths to all the keywords (of the query
$K=\left\{{k_1,\ldots,k_n}\right\}$), a node has a separate mark for
each keyword.  The marks of a node $v$ are stored in the array
$v.\mathit{marks}$, which has an entry for each keyword.  For a
keyword $k_i$, the mark of $v$ (i.e.,~$v.\mathit{marks[k_i]}$) means
the following.  Node $v$ is $\act$ if we have not yet discovered that
there is a path from $v$ to $k_i$.  Node $v$ is $\vis$ if a minimal
path from $v$ to $k_i$ has been produced.  And $v$ is marked as
$\inan$ when we discover for the first time that $v$ is on a path
from some $K$-root to $k_i$.

If $v.\mathit{marks[k_i]}$ is $\vis$ and a path $p$ from $v$ to $k_i$
is removed from the queue, then $p$ is \emph{frozen} at
$v$.  Frozen paths from $v$ to $k_i$ are stored in a
dedicated list $v.\mathit{frozen}[k_i]$.  The paths of
$v.\mathit{frozen}[k_i]$ are \emph{unfrozen} (i.e.,~are moved back into
the queue) when $v.\mathit{marks[k_i]}$ is changed to $\inan$.

We now describe the execution of GTF on the graph snippet of
Figure~\ref{fig:datagraph}, assuming that the query is
$K=\left\{\mathit{France}, \mathit{Paris}\right\}$. 
Initially, two paths $\langle\mathit{France}\rangle$ and 
$\langle\mathit{Paris}\rangle$,
each consisting of one keyword of $K$, are inserted into the queue,
where lower weight means higher priority.
Next, the top of the queue is removed; suppose that it is 
$\langle\mathit{France}\rangle$.
First, we change $\mathit{France}.\mathit{marks}[\mathit{France}]$ to $\vis$.
Second, for each parent $v$ of $\mathit{France}$,
the path $v \rightarrow \mathit{France}$ is inserted into the queue;
namely, these are the paths $p_1$ and $p_2$ of Figure~\ref{fig:gt}.
We continue to iterate in this way. Suppose that now 
$\langle\mathit{Paris}\rangle$ has the lowest weight.
So, it is removed from the queue,
$\mathit{Paris}.\mathit{marks}[\mathit{Paris}]$ is changed to $\vis$,
and the path $p_7$ (of Figure~\ref{fig:gt}) is inserted into the queue.

Now, let the path $p_1$ be removed from the queue.
As a result, $\mathit{province}.\mathit{marks}[\mathit{France}]$ is
changed to $\vis$, and the path $p_6= \mathit{city} \rightarrow p_1$
is inserted into the queue.  Next, assume that $p_2$ is removed from
the queue.  So, $\mathit{country}.\mathit{marks}[\mathit{France}]$ is
changed to $\vis$, and the paths $p_3= \mathit{province} \rightarrow
p_2$ and $p_5= \mathit{city} \rightarrow p_2$ are inserted into the queue.

Now, suppose that $p_3$ is at the top of the queue.
So, $p_3$ is removed and immediately frozen at $\mathit{province}$
(i.e.,~added to $\mathit{province}.\mathit{frozen}[\mathit{France}]$),
because $\mathit{province}.\mathit{marks}[\mathit{France}]=\vis$.
Consequently, no paths are added to the queue in this iteration.
Next, assume that $p_6$ is removed from the queue.
The value of $\mathit{city}.\mathit{marks}[\mathit{France}]$ is changed to 
$\vis$ and no paths are inserted into the queue,
because $\mathit{city}$ has no incoming edges.

Now, suppose that $p_7$ is at the top of the queue.  So, it is removed
and $\mathit{city}.\mathit{marks}[\mathit{Paris}]$ is changed to
$\vis$.  Currently, both
$\mathit{city}.\mathit{marks}[\mathit{Paris}]$ and
$\mathit{city}.\mathit{marks}[\mathit{France}]$ are $\vis$.  That is,
there is a path from $\mathit{city}$ to all the keywords of the query
$\left\{\mathit{France}, \mathit{Paris}\right\}$.  Recall that the
paths that have reached $\mathit{city}$ so far are $p_6$ and $p_7$.
For each one of those paths $p$, the following is done,
assuming that $p$ ends at the keyword $k$.
For each node $v$ of $p$, we change the mark of $v$ for $k$ to $\inan$ 
and unfreeze paths to $k$ that are frozen at $v$.
Doing it for $p_6$ means that
$\mathit{city}.\mathit{marks}[\mathit{France}]$,
$\mathit{province}.\mathit{marks}[\mathit{France}]$ and
$\mathit{France}.\mathit{marks}[\mathit{France}]$ are all changed to
$\inan$.  In addition, the path $p_3$ is removed from
$\mathit{province}.\mathit{frozen}[\mathit{France}]$ and inserted back
into the queue.  We act similarly on $p_7$. That is,
$\mathit{city}.\mathit{marks}[\mathit{Paris}]$ and
$\mathit{Paris}.\mathit{marks}[\mathit{Paris}]$ are changed to
$\inan$.  In this case, there are no paths to be unfrozen.

Now, the marks of $\mathit{city}$ for all the keywords (of the query)
are $\inan$. Hence, we generate answers from the paths that have
already reached $\mathit{city}$.  As a result, the answer $A_1$ of
Figure~\ref{fig:answers} is produced.  Moreover, from now on, when a
new path reaches $\mathit{city}$, we will try to generate more answers
by applying $\prodAns(\mathcal{P}, p)$.

\subsection{The Pseudocode of the GTF Algorithm}\label{sec:pseudocode}
\begin{figure}[!h]
\hrule
\begin{tabular}{l l}
	{\bf Algorithm:} & \textit{GTF (Generate Trees with Freezing)} \\
	{\bf Input:}     & $G=(V,E)$ is a data graph \\
			 		 & $K$ is a set of keyword nodes \\
	{\bf Output:}    & Answers to $K$ \\
\end{tabular}
\hrule
\begin{algorithmic}[1]
  \STATE $Q\gets$ an empty priority queue 		\label{alg:gtf:initQ_start}
  \FOR{$v\in V$}
     \STATE $v.\mathit{isKRoot}\gets \false$
  \ENDFOR
  \FOR{$v\in V$ and $k\in K$}
     \STATE $v.\mathit{paths[k]}\gets \emptyset$
     \STATE $v.\mathit{frozen[k]\gets\emptyset}$  \label{alg:gtf:initializeFrozen}
     \STATE $v.\mathit{marks[k]}\gets\act$        \label{alg:gtf:initializeMarks}
  \ENDFOR
  \FOR{$k\in K$} 
  	\STATE $Q.\pushq(\langle k \rangle)$
  \ENDFOR										\label{alg:gtf:initQ_end}
  \WHILE{$Q$ is not empty}						\label{alg:gtf:mainLoop_start}
  	\STATE $p\gets Q.\popq()$			\label{alg:gtf:popBestPath}
	\IF{\textbf{freeze}($p$)}					\label{alg:gtf:testFreeze}
  	  \STATE \continue				\label{alg:gtf:freezeContinue}
    \ENDIF
  	\STATE $v\gets\firstp(p)$
  	\IF{$v.\mathit{marks}[p.\pkwd]=\act$ }			\label{alg:gtf:testActiveStart}
  	  \STATE $v.\mathit{marks}[p.\pkwd]\gets\vis$ 	\label{alg:gtf:markVisited}
  	\ENDIF													\label{alg:gtf:testActiveEnd}
  	\STATE $\mathit{relax}\gets \textbf{true}$ \label{alg:gtf:relaxGetsTrue}
  	\IF{$v.\mathit{isKRoot} = \true$}					\label{alg:gtf:testOldRoot}
		\STATE $\textbf{unfreeze}(p,Q)$				\label{alg:gtf:unfrAtOldRoot}
  		\IF{$p$ has no cycles}							\label{alg:gtf:testLoopness}
		    \STATE $v.\mathit{paths}[p.\pkwd].\textbf{add}(p)$	\label{alg:gtf:addNodeVisit1}
  			\STATE $\prodAns(v.paths,p)$							\label{alg:gtf:prodAns1}
  		\ELSE
  			\STATE $\mathit{relax}\gets \false$ \label{alg:gtf:reassignFalse}
	    \ENDIF
	\ELSE
		\STATE $v.\mathit{paths}[p.\pkwd].\textbf{add}(p)$ \label{alg:gtf:addToPaths}
		\IF {for all $k\in K$, it holds that $v.\mathit{paths[k]}\not= \emptyset$} \label{alg:gtf:becomesRoot}
			\STATE $v.\mathit{isKRoot}\gets \true$ \label{alg:gtf:isRootGetsTrue}
			\FOR{$k \in K$}								\label{alg:gtf:newRootStart}
				\FOR{$p' \in v.\mathit{pahts}[k]$}		
			  	    \STATE $\textbf{unfreeze}(p',Q)$	\label{alg:gtf:unfrNewRoot}
			  	\ENDFOR
                           \STATE remove cyclic paths from $v.\mathit{paths}[k]$ \label{alg:gtf:removeCyclic}
                        \ENDFOR \label{alg:gtf:newRootEnd}
	  		   \STATE $\prodAns(v.paths,p)$ \label{alg:gtf:prodAns2}
  		\ENDIF	  	  	
  	\ENDIF
  	\IF{$\mathit{relax}$} \label{alg:gtf:relaxIsTrue}
	  	\FOR{$v' \in \parents(v)$}						\label{alg:gtf:relaxStart}
	        \IF{$v'$ is not on $p$ or $v'\rightarrow p$ is essential} \label{alg:gtf:testEssential}
	  			\STATE $Q.\pushq(v' \rightarrow p)$ \label{alg:gtf:insertQ}
	  		\ENDIF
	  	\ENDFOR	
  	\ENDIF										\label{alg:gtf:relaxEnd}
  \ENDWHILE										\label{alg:gtf:mainLoop_end}
\end{algorithmic}
\hrule
\caption{The GTF algorithm\label{alg:gtf}}
\end{figure}

\begin{figure}[t]
\begin{minipage}{0.53\textwidth}
\hrule
\begin{tabular}{l l}
	{\bf Procedure:} & $\textbf{freeze}(p)$ \\
\end{tabular}
\hrule
\begin{algorithmic}[1]
\IF{$\firstp(p).\mathit{marks}[p.\pkwd] = \vis$}
	\STATE $\firstp(p).\mathit{frozen}[p.\pkwd].\textbf{add}(p)$
	\STATE \textbf{return} \true
\ELSE	
	\STATE \textbf{return} \false
\ENDIF									\label{alg:gtf:freeze:lastLine}
\end{algorithmic}
\hrule
\begin{tabular}{l l}
	{\bf Procedure:} & $\textbf{unfreeze}(p,Q)$ \\
\end{tabular}
\hrule
\begin{algorithmic}[1]
\STATE $p'\gets p$
\WHILE{$p' \neq \perp$}
	\STATE ${\bar v}\gets \firstp(p')$
	\IF{${\bar v}.\mathit{marks}[p.\pkwd] \neq \inan$}
		\STATE ${\bar v}.\mathit{marks}[p.\pkwd]\gets\inan$
		\FOR{$p'' \in {\bar v}.\mathit{frozen}[p.\pkwd]$}
			\STATE $Q.\pushq(p'')$ \label{unfreeze:insert}
		\ENDFOR
	\ENDIF
	\STATE $p'\gets \predp(p')$
\ENDWHILE
\end{algorithmic}
\end{minipage}
\hfill
\begin{minipage}{0.45\textwidth}
\hrule
\begin{tabular}{l l}
	{\bf Procedure:} & $\prodAns(\mathcal{P}, p)$ \\
	{\bf Output:}    & answers rooted at $\firstp(p)$
\end{tabular}
\hrule
\begin{algorithmic}[1]
\STATE $\mathcal{P}[p.\pkwd]\gets\left\{{p}\right\}$  \label{alg:gt:prodAns:changeP}
\STATE $\mathit{iter}\gets$ new $\itgrs(\mathcal{P})$     \label{alg:ft:prodAns:newIter}
\WHILE {$\mathit{iter}.\textbf{hasNext}()$} \label{alg:gt:prodAns:mainLoop}
	\STATE $\bar{P}\gets \mathit{iter}.\nexti()$ 			\label{alg:gt:prodAns:next}	
	\STATE $a\gets$ combine all the paths in $\bar{P}$ \label{alg:gt:prodAns:createTree}
	\\ \emph{/*} $\nexti()$ \emph{ensures that the combination of all the paths in $\bar{P}$ yields a tree (rather than a graph) */}
	\IF {the root of $a$ has more than one child} 		\label{alg:gt:prodAns:testRoot}
		\STATE \textbf{output} $a$  \label{alg:gt:prodAns:output}
	\ENDIF
\ENDWHILE
\end{algorithmic}
\hrule
\end{minipage}
\caption{Helper procedures for the GTF algorithm}
\label{alg:gtf:helpers}
\end{figure}

The GTF algorithm is presented in Figure~\ref{alg:gtf}
and its helper procedures---in Figure~\ref{alg:gtf:helpers}. The input is
a data graph $G=(V,E)$ and a query $K=\left\{{k_1,\ldots,k_n}\right\}$.
The algorithm uses a single priority queue $Q$ to generate,
by increasing weight, all simple paths to every keyword node of $K$.
For each node $v\in V$, there is a flag $\mathit{isKRoot}$ that
indicates whether $v$ has a path to each keyword of $K$.
Initially, that flag is \textbf{false}.
For each node $v\in V$, the set of the constructed paths from $v$ to
the keyword $k$ is stored in $v.paths[k]$, which is initially empty.
Also, for all the keywords of $K$ and nodes of $G$, we initialize
the marks to be $\act$ and the lists of frozen paths to be empty.
The paths are constructed backwards, that is, from the last node (which is always a keyword).  
Therefore, for each $k\in K$, we insert the path
$\langle k \rangle$ (consisting of the single node $k$) into $Q$.
All these initializations are done in lines \ref{alg:gtf:initQ_start}--\ref{alg:gtf:initQ_end} (of Figure~\ref{alg:gtf}).

The main loop of
lines~\ref{alg:gtf:mainLoop_start}--\ref{alg:gtf:insertQ} is repeated
while $Q$ is not empty.  Line~\ref{alg:gtf:popBestPath} removes the
best (i.e.,~least-weight) path $p$ from $Q$.  Let $v$ and $k_i$ be the first
and last, respectively, nodes of $p$.  Line~\ref{alg:gtf:testFreeze}
freezes $p$ provided that it has to be done.  This is accomplished by
calling the procedure $\textbf{freeze}(p)$ of
Figure~\ref{alg:gtf:helpers} that operates as follows.  If the mark of
$v$ for $k_i$ is $\vis$, then $p$ is frozen at $v$ by adding it to
$v.frozen[k_i]$ and \textbf{true} is returned; in addition, the main
loop continues (in line~\ref{alg:gtf:freezeContinue}) to the next
iteration.  Otherwise, \textbf{false} is returned and $p$ is handled
as we describe next.

Line~\ref{alg:gtf:testActiveStart} checks if $p$ is the first path
from $v$ to $k_i$ that has been removed from $Q$.
If so, line~\ref{alg:gtf:testActiveEnd}
changes the mark of $v$ for $k_i$ from $\act$ to $\vis$.
Line~\ref{alg:gtf:relaxGetsTrue} assigns \textbf{true} to the flag
$\mathit{relax}$, which means that (as of now) $p$ should
spawn new paths that will be added to $Q$.

The test of line~\ref{alg:gtf:testOldRoot} splits the execution of
the algorithm into two cases.
If $v$ is a $K$-root (which must have
been discovered in a previous iteration and means that for every
$k\in K$, there is a path from $v$ to $k$), then the
following is done.  First, line~\ref{alg:gtf:unfrAtOldRoot} calls the
procedure $\textbf{unfreeze}(p,Q)$ of Figure~\ref{alg:gtf:helpers}
that unfreezes (i.e.,~inserts into $Q$) all the paths to $k_i$ that
are frozen at nodes of $p$ (i.e.,~the paths of ${\bar v}.frozen[k_i]$,
where $\bar v$ is a node of $p$).  In addition, for all nodes ${\bar v}$
of $p$, the procedure $\textbf{unfreeze}(p,Q)$ changes the
mark of ${\bar v}$ for $k_i$ to $\inan$.
Second, line~\ref{alg:gtf:testLoopness} tests whether $p$ is acyclic. If so,
line~\ref{alg:gtf:addNodeVisit1} adds $p$ to the paths of $v$ that reach $k_i$,
and line~\ref{alg:gtf:prodAns1} produces new answers that include $p$
by calling $\textbf{produceAnswers}$ of Figure~\ref{alg:gtf:helpers}.
The pseudocode of $\prodAns(v.paths,p)$ is just an efficient implementation
of considering every combination of paths $p_1,\ldots,p_n$, 
such that $p_i$ is from $v$ to $k_i$ ($1\le i\le n$),
and checking that it is an answer to $K$.
(It should be noted that GTF generates answers by increasing height.)
If the test of line~\ref{alg:gtf:testLoopness} is \textbf{false},
then the flag $\mathit{relax}$ is changed back to \textbf{false},
thereby ending the current iteration of the main loop.

If the test of line~\ref{alg:gtf:testOldRoot} is \textbf{false}
(i.e.,~$v$ has not yet been discovered to be a $K$-root), the execution
continues in line~\ref{alg:gtf:addToPaths} that adds $p$ to the paths
of $v$ that reach $k_i$.  Line~\ref{alg:gtf:becomesRoot} tests whether
$v$ is now a $K$-root and if so, the flag $\mathit{isKRoot}$ is set to
\textbf{true} and the following is done.  The nested loops of
lines~\ref{alg:gtf:newRootStart}--\ref{alg:gtf:newRootEnd} iterate
over all paths $p'$ (that have already been discovered) from $v$ to
any keyword node of $K$ (i.e.,~not just $k_i$).  For each $p'$, where
${k'}$ is the last node of $p'$ (and, hence, is a keyword),
line~\ref{alg:gtf:unfrNewRoot} calls $\textbf{unfreeze}(p',Q)$,
thereby inserting into $Q$ all the paths to ${k'}$ that are frozen
at nodes of $p'$ and changing the mark (for ${k'}$) of those nodes
to $\inan$.
Line~\ref{alg:gtf:removeCyclic} removes all the cyclic paths among
those stored at $v$.  Line~\ref{alg:gtf:prodAns2} generates answers
from the paths that remain at $v$.

If the test of line~\ref{alg:gtf:relaxIsTrue} is \textbf{true},
the relaxation of $p$ is done in
lines~\ref{alg:gtf:relaxStart}--\ref{alg:gtf:relaxEnd} as follows.
For each parent $v'$ of $v$, the
path $v' \rightarrow p$ is inserted into $Q$ if either one of the
following two holds (as tested in
line~\ref{alg:gtf:testEssential}). First, $v'$ is not on $p$.
Second, $v' \rightarrow p$
is essential, according to the following definition.  The path $v'
\rightarrow p$ is \emph{essential} if $v'$ appears on $p$ and the
section of $v' \rightarrow p$ from its first node (which is $v'$) to
the next occurrence of $v'$ has at least one node $u$, such that
$u.\mathit{marks[k]}=\vis$, where the keyword $k$ is the last node of $p$.
Appendix~\ref{sec:looping_paths} 
gives an example
that shows why essential paths (which are cyclic)
have to be inserted into $Q$.

Note that due to line~\ref{alg:gtf:reassignFalse}, no cyclic path $p[v,k]$
is relaxed if $v$ has already been discovered to
be a $K$-root in a previous iteration. The reason is that none of the nodes
along $p[v,k]$ could have the mark $\vis$ for the keyword $k$ (hence, 
no paths are frozen at those nodes).

Observe that before $v$ is known to be a $K$-root, we add cyclic paths to
the array $v.paths$. Only when discovering that $v$ is a $K$-root, do we
remove all cyclic paths from $v.paths$ (in
line~\ref{alg:gtf:removeCyclic}) and stop adding them in subsequent
iterations.  This is lazy evaluation, because prior to
knowing that answers with the $K$-root $v$ should be produced, it is a
waste of time to test whether paths from $v$ are cyclic.
\section{Correctness and Complexity of GTF} \label{correctness}
\subsection{Definitions and Observations}
Before proving correctness of the GTF algorithm,
we define some notation and terminology (in addition 
to those of Section~\ref{sec:prelim}) and state a few observations.
Recall that the data graph is $G=(V,E)$.
Usually, a keyword is denoted by $k$, whereas $r$, $u$, $v$ and $z$
are any nodes of $V$.

We only consider directed paths of $G$ that are defined as usual.
If $p$ is a path from $v$ to $k$, then we write it as $p[v,k]$ when we want
to explicitly state its first and last nodes.
We say that node $u$ is \emph{reachable} from $v$ if there is a path
from $v$ to $u$.

A \emph{suffix} of $p[v,k]$ is a traversal of $p[v,k]$ that starts at 
(some particular occurrence of) a node $u$ and ends at the last node of $p$.
Hence, a suffix of $p[v,k]$ is denoted by $p[u,k]$.
A \emph{prefix} of $p[v,k]$ is a traversal of $p[v,k]$ that starts at $v$
and ends at (some particular occurrence of) a node $u$. Hence, a prefix of
$p[v,k]$ is denoted by $p[v,u]$.
A suffix or prefix of $p[v,k]$ is \emph{proper} if it is different from
$p[v,k]$ itself.

Consider two paths $p_1[v,z]$ and $p_2[z,u]$; that is, the former ends
in the node where the latter starts. Their \emph{concatenation},
denoted by $p_1[v,z] \concat p_2[z,u]$, is obtained by joining them at node $z$.

As already mentioned in Section~\ref{sec:prelim},
a positive \emph{weight} function $w$ is defined on the
nodes and edges of $\G$. The weight of a path $p[v,u]$, denoted by
$\fweight(p[v,u])$, is the sum of weights over all the nodes and
edges of $p[v,u]$. A \emph{minimal} path from $v$ to $u$ has the minimum weight
among all paths from $v$ to $u$.  Since the weight function is
positive, there are no zero-weight cycles.  Therefore, a minimal path
is acyclic.  Also observe that the weight of a proper suffix or prefix
is strictly smaller than that of the whole path.\footnote{For the proof of correctness, it is
  enough for the weight function to be non-negative
  (rather than positive) provided that every cycle has a positive weight.
}

Let $K$ be a query (i.e.,~a set of at least two keywords).  Recall
from Section~\ref{sec:prelim} the definitions of $K$-root, $K$-subtree
and height of a subtree.  The \emph{best height} of a $K$-root $r$ is
the maximum weight among all the minimal paths from $r$ to any keyword
$k\in K$.  Note that the height of any $K$-subtree rooted at $r$ is at
least the best height of $r$.

Consider a nonempty set of nodes $S$ and a node $v$.
If $v$ is reachable from every node of $S$, then we say that node $u\in S$
is \emph{closest to} $v$ if a minimal path from $u$ to $v$ has the 
minimum weight among all paths from any node of $S$ to $v$.

Similarly, if every node of $S$ is reachable from $v$, then we say
that node $u\in S$ is \emph{closest from} $v$ if a minimal path from
$v$ to $u$ has the minimum weight among all paths from $v$ to any node
of $S$.

In the sequel, line numbers refer to the algorithm GTF of Figure~\ref{alg:gtf},
unless explicitly stated otherwise.
We say that a node $v \in \V$ is \emph{discovered as a $K$-root}
if the test of line~\ref{alg:gtf:becomesRoot} is satisfied
and $v.\mathit{isRoot}$ is assigned $\true$ in 
line~\ref{alg:gtf:isRootGetsTrue}.
Observe that the test of line~\ref{alg:gtf:becomesRoot} is $\true$
if and only if for all $k\in K$, it holds that 
$v.\mathit{marks[k]}$ is either $\vis$ or $\inan$.
Also note that line~\ref{alg:gtf:isRootGetsTrue} is executed at most once
for each node $v$ of $\G$. Thus, there is at most one iteration of the main
loop (i.e.,~line~\ref{alg:gtf:mainLoop_start}) that discovers $v$ as $K$-root.

We say that a path $p$ is \emph{constructed} when it is inserted into $Q$
for the first time, which must happen in line~\ref{alg:gtf:insertQ}.
A path is \emph{exposed} when it is removed from $Q$ in 
line~\ref{alg:gtf:popBestPath}. Observe that a path $p[v,k]$ may be exposed
more than once, due to freezing and unfreezing.

\theoremstyle{definition}
\newtheorem{proposition}[theorem]{Proposition}
\begin{proposition}\label{prop:twice}
A path can be exposed at most twice.
\end{proposition}
\begin{proof}
When an iteration exposes a path $p[v,k]$ for the first time,
it does exactly one of the following.
It freezes $p[v,k]$ at node $v$,
discard $p[v,k]$ due to line~\ref{alg:gtf:reassignFalse}, or
extend (i.e.,~relax) $p[v,k]$ in the loop of line~\ref{alg:gtf:relaxStart}
and inserts the results into $Q$ in line~\ref{alg:gtf:insertQ}.
Note that some relaxations of $p[v,k]$ are never inserted into $Q$,
due to the test of line~\ref{alg:gtf:testEssential}.
Only if $p[v,k]$ is frozen at $v$, can it be inserted a second time
into $Q$, in line~\ref{unfreeze:insert} of the procedure $\textbf{unfreeze}$
(Figure~\ref{alg:gtf:helpers}) that also sets $v.\mathit{marks[k]}$ to $\inan$.
But then $p[v,k]$ cannot freeze again at $v$,
because $v.\mathit{marks[k]}$ does not change after becoming $\inan$.
Therefore, $p[v,k]$ cannot be inserted into $Q$ a third time.
\end{proof}

In the next section, we sometimes refer to the mark of
a node $v$ of a path $p$.  It should be clear from the context that we
mean the mark of $v$ for the keyword where $p$ ends.

\subsection{The Proof}
We start with an auxiliary lemma that considers the concatenation of
two paths, where the linking node is $z$, as shown in
Figure~\ref{fig:illustration_a} (note that a wavy arrow denotes a
path, rather than a single edge). Such a concatenation is used in the
proofs of subsequent lemmas. 

\tikzstyle{vertex}=[ellipse,draw,minimum size=14pt,inner sep=1pt]   
\tikzstyle{keyword}=[rectangle,draw,minimum size=14pt,inner sep=1pt]  

\tikzstyle{path}=[decorate, decoration = {snake, post length=3pt, segment length=4pt,amplitude=0.6pt},->,>=stealth']
\tikzstyle{path-}=[decorate, decoration = {snake, post length=0pt, segment length=4pt,amplitude=0.6pt},-,>=stealth']
\tikzset{every edge/.append style={->,>=stealth',shorten >=0pt}}

\begin{figure}
\centering
\begin{tikzpicture}
  
  \node[vertex] (v) at (0,0) {$v$};
  \node[vertex] (u) at (2,0.33){$u$};
  \node[vertex] (z) at (3,0.5){$z$};
  \node[keyword] (k) at (6,0){$k$};
  \node (text1) at (1.5,-0.5){$p_s[v,z]$};
  \node (text2) at (4.5,-0.5){$p_m[z,k]$};
   
  \draw  (v)  edge [path] (u); 
  \draw  (u)  edge [path] (z); 
  \draw  (z)  edge [path] (k);
  
\end{tikzpicture}
\caption{The path $\bar{p}[v,k]$}
\label{fig:illustration_a}
\end{figure}
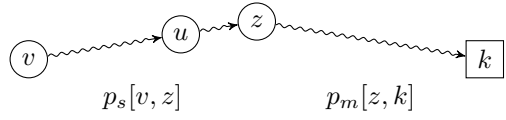

\begin{restatable}{lemma}{lemmaPC}
\label{LEMMA:PATH-CONCAT}
  Let $k$ be a keyword of the query $K$, and let $v$ and $z$ be nodes of the
  data graph.  Consider two paths $p_s[v,z]$ and $p_m[z,k]$.
  Let $\bar p[v,k]$ be their concatenation at node $z$, that is,
\begin{equation*}
\bar p[v,k] =  p_s[v,z] \concat p_m[z,k].  
\end{equation*}
Suppose that the following hold at the beginning of iteration $i$ of the 
main loop (line~\ref{alg:gtf:mainLoop_start}).
\begin{enumerate}
\item\label{cond:new-first}
The path $p_s[v,z]$ is minimal or (at least) acyclic.
\item\label{cond:new-second}
The path $p_m[z,k]$ has changed $z.\mathit{marks[k]}$ from $\act$ to $\vis$
in an earlier iteration.
\item\label{cond:first}
$z.\mathit{marks[k]}=\vis$.
\item\label{cond:second}
For all nodes $u\not=z$ on the path $p_s[v,z]$, 
the suffix $\bar p[u,k]$ is not frozen at $u$.
\item\label{cond:third}
The path $\bar p[v,k]$ has not yet been exposed.
\end{enumerate}
Then, some suffix of $\bar p[v,k]$ must be on $Q$ at the beginning of
iteration $i$.
\end{restatable}

\begin{proof}
  Suppose, by way of contradiction, that no suffix of $\bar p[v,k]$ is
  on $Q$ at the beginning of iteration $i$. Since $\bar p[v,k]$
  has not yet been exposed, there are two possible cases regarding its state.
  We derive a contradiction by showing that none of them can happen.
\begin{description}
\item[Case 1:] Some suffix of $\bar p[v,k]$ is frozen.  This cannot
  happen at any node of $\bar p[z,k]$ (which is the same as
  $p_m[z,k]$), because Condition~\ref{cond:first} implies that
  $p_m[z,k]$ has already changed $z.\mathit{marks[k]}$ to $\vis$.
  Condition~\ref{cond:second} implies that it cannot happen at the
  other nodes of $\bar p[v,k]$ (i.e.,~the nodes $u$ of $p_s[v,z]$ that
  are different from $z$).
\item[Case 2:] Some suffix of $\bar p[v,k]$ has already been discarded 
(in an earlier iteration) either by the test of line~\ref{alg:gtf:testEssential}
or due to line~\ref{alg:gtf:reassignFalse}.
  This cannot happen to any
  suffix of $\bar p[z,k]$ (which is the same as $p_m[z,k]$), because
  $p_m[z,k]$ has already changed $z.\mathit{marks[k]}$ to $\vis$. We
  now show that it cannot happen to any other suffix $\bar p[u,k]$,
  where $u$ is a node of $p_s[v,z]$ other than $z$.  Note that $\bar
  p[v,k]$ (and hence $\bar p[u,k]$) is not necessarily
  acyclic. However, the lemma states that $p_s[v,z]$ is
  acyclic. Therefore, if the suffix $\bar p[u,k]$, has a cycle that
  includes $u$, then it must also include $z$.  But
  $z.\mathit{marks[k]}$ is $\vis$ from the moment it was changed to
  that value until the beginning of iteration $i$ (because a mark
  cannot be changed to $\vis$ more than once). Hence, the suffix $\bar
  p[u,k]$ could not have been discarded by the test of
  line~\ref{alg:gtf:testEssential}.
It is also not possible that
line~\ref{alg:gtf:reassignFalse} 
has already discarded $\bar p[u,k]$ for the following reason.
If line~\ref{alg:gtf:reassignFalse} is reached
(in an iteration that removed $\bar p[u,k]$ from $Q$), then 
for all nodes $x$ on $\bar p[u,k]$,  line~\ref{alg:gtf:unfrAtOldRoot}
has already changed $x.\mathit{marks[k]}$ to $\inan$.
Therefore, $z.\mathit{marks[k]}$ cannot be $\vis$ at the beginning
of iteration $i$. 

\end{description}
It thus follows that some suffix of $\bar p[v,k]$ is on $Q$
at the beginning of iteration $i$.
\end{proof}

\begin{restatable}{lemma}{lemmaSPM}
\label{LEMMA:GTF-SHORTEST-PATH-MARKS} 
For all nodes $v\in \V$ and keywords $k\in K$, 
the mark $v.\mathit{marks[k]}$ can be changed from $\act$ to $\vis$
only by a minimal path from $v$ to $k$.
\end{restatable}

\begin{proof}
Suppose that the lemma is not true for some keyword $k\in K$.
Let $v$ be a closest node to $k$
among all those violating
the lemma with respect to $k$.
Node $v$ is different from $k$, because the path
$\anset{k}$ marks $k$ as $\vis$. We will derive a contradiction by
showing that a minimal path changes $v.\mathit{marks[k]}$ 
from $\act$ to $\vis$.
 
Let $p_s[v,k]$ be a minimal path  from $v$ to $k$. 
Consider the iteration $i$ of the main loop
(line~\ref{alg:gtf:mainLoop_start} in Figure~\ref{alg:gtf}) that changes 
$v.\mathit{marks[k]}$ to $\vis$ (in line~\ref{alg:gtf:markVisited}).
Among all the nodes of $p_s[v,k]$
in which suffixes of some minimal paths from $v$ to $k$ are
frozen at the beginning of iteration $i$, let $z$ be the first one
when traversing $p_s[v,k]$ from $v$ to $k$
(i.e.,~on the path $p_s[v,z]$, node $z$ is the only one in which
such a suffix is frozen).  Node $z$ exists for the following three reasons.
\begin{itemize}
\item The path $p_s[v,k]$ has not been exposed prior to
  iteration $i$, because we assume that $v.\mathit{marks[k]}$ is
  changed to $\vis$ in iteration $i$ and that change can happen only once.
\item The path $p_s[v,k]$ is acyclic (because it is minimal), so
  a suffix of $p_s[v,k]$ could not have been discarded either
  by the test of line~\ref{alg:gtf:testEssential} or due to 
line~\ref{alg:gtf:reassignFalse}.
\item The path $p_s[v,k]$ (or any suffix thereof) cannot be on the
  queue at the beginning of iteration $i$, because $v$ violates the
  lemma, which means that a non-minimal path from $v$ to $k$ must be
  removed from the queue at the beginning of that iteration.
\end{itemize}
The above three observations imply that a proper suffix of $p_s[v,k]$
must be frozen at the beginning of iteration $i$ and, hence, node $z$ exists.
Observe that $z$ is different from $v$, because a path to $k$
can be frozen only at a node $\hat v$, such that 
${\hat v}.\mathit{marks[k]}=\vis$, whereas we assume that
$v.\mathit{marks[k]}$ is $\act$ at the beginning of iteration $i$.

By the selection of $v$ and $p_s[v,k]$ (and the above fact that $z\not=v$),
node $z$ does not violate the lemma, because $p_s[z,k]$ is a proper
suffix of $p_s[v,k]$ and, hence, $z$ is closer to $k$ than $v$.
Therefore, according to the lemma, there is a minimal path $p_m[z,k]$
that changes $z.\mathit{marks[k]}$ to $\vis$.  Consequently,
  \begin{equation}\label{eqn:min-path}
  \fweight(p_m[z,k]) \leq \fweight(p_s[z,k]).
  \end{equation}
Now, consider the path 
  \begin{equation}\label{eqn:new-path}
  \bar{p}[v,k] = p_s[v,z] \concat p_m[z,k].
  \end{equation}
  Since $p_s[v,k]$ is a minimal path from $v$ to $k$,
  Equations~(\ref{eqn:min-path}) and~(\ref{eqn:new-path}) imply that
  so is $\bar{p}[v,k]$.  

  We now show that the conditions of Lemma~\ref{LEMMA:PATH-CONCAT} are
  satisfied at the beginning of iteration $i$.  In particular,
  Condition~\ref{cond:new-first} holds, because $p_s[v,k]$ is acyclic
  (since it is minimal) and, hence, so is the path $p_s[v,z]$.
  Condition~\ref{cond:new-second} is satisfied, because of how
  $p_m[z,k]$ is defined.  Condition~\ref{cond:first} holds, because we chose 
  $z$ to be a node where a path to $k$ is frozen.
  Condition~\ref{cond:second} is satisfied, because of how $z$ was
  chosen and the fact that $\bar{p}[v,k]$ is minimal.
  Condition~\ref{cond:third} is satisfied, because we have assumed
  that $v.\mathit{marks[k]}$ is changed from $\act$ to $\vis$ during
  iteration $i$.

By Lemma~\ref{LEMMA:PATH-CONCAT}, a suffix of
$\bar p[v,k]$ must be on the queue at the beginning of iteration
$i$. This contradicts our assumption that a non-minimal path (which
has a strictly higher weight than any suffix of $\bar p[v,k]$) changes
$v.\mathit{marks[k]}$ from $\act$ to $\vis$ in iteration $i$.
\end{proof}

\begin{restatable}{lemma}{lemmaOQ}
\label{LEMMA:ON-QUEUE}
For all nodes $v\in \V$ and keywords $k\in K$, such that $k$ is reachable
from $v$, if $v.\mathit{marks[k]}$ is $\act$ at the beginning of an iteration
of the main loop (line~\ref{alg:gtf:mainLoop_start}), then $Q$ contains a suffix
(which is not necessarily proper) of a minimal path from $v$ to $k$.
\end{restatable}

\begin{proof}
The lemma is certainly true at the beginning of the first iteration,
because the path $\anset{k}$ is on $Q$.
Suppose that the lemma does not hold at the beginning of iteration $i$.
Thus, every minimal path $p[v,k]$ has a proper suffix that is frozen 
at the beginning of iteration $i$.
(Note that a suffix of a minimal path cannot be discarded either by the 
test of line~\ref{alg:gtf:testEssential} or due to 
line~\ref{alg:gtf:reassignFalse}, because it is acyclic.)
Let $z$ be the closest node from $v$ having such a frozen suffix.
Hence, $z.\mathit{marks[k]}$ is $\vis$ and $z\not= v$
(because $v.\mathit{marks[k]}$ is $\act$).
By Lemma~\ref{LEMMA:GTF-SHORTEST-PATH-MARKS},
a minimal path $p_m[z,k]$ has changed $z.\mathit{marks[k]}$ to $\vis$.
Let $p_s[v,z]$ be a minimal path from $v$ to $z$.
Consider the path
\begin{equation*} 
\bar{p}[v,k] = p_s[v,z] \concat p_m[z,k].
\end{equation*}
The weight of $\bar{p}[v,k]$ is no more than that of a minimal path
from $v$ to $k$, because both $p_s[v,z]$ and $p_m[v,k]$ are minimal
and the choice of $z$ implies that it is on some minimal path from $v$
to $k$.  Hence, $\bar{p}[v,k]$ is a minimal path from $v$ to $k$.

We now show that the conditions of Lemma~\ref{LEMMA:PATH-CONCAT} are satisfied.
Conditions~\ref{cond:new-first}--\ref{cond:first} clearly hold.
Condition~\ref{cond:second} is satisfied because of how $z$ is chosen
and the fact that $\bar{p}[v,k]$ is minimal.
Condition~\ref{cond:third} holds because $v.\mathit{marks[k]}$ is $\act$
at the beginning of iteration $i$.

By Lemma~\ref{LEMMA:PATH-CONCAT}, a suffix of $\bar{p}[v,k]$ is on $Q$
at the beginning of iteration $i$, contradicting our initial assumption.
\end{proof}

\begin{lemma}\label{lemma:finite}
Any constructed path can have at most $2n(n+1)$ nodes, where $n=|V|$ 
(i.e.,~the number of nodes in the graph).
Hence, the algorithm constructs at most $(n+1)^{2n(n+1)}$ paths. 
\end{lemma}
\begin{proof}
  We say that $v_m \rightarrow \cdots \rightarrow v_1$ is a
  \emph{repeated run} in a path $\bar p$ if some suffix (not
  necessarily proper) of $\bar p$ has the form $v_m \rightarrow \cdots
  \rightarrow v_1 \rightarrow p$, where each $v_i$ also appears in any
  two positions of $p$.  In other words, for all $i$ ($1\le i \le m$),
  the occurrence of $v_i$ in $v_m \rightarrow \cdots \rightarrow v_1$
  is (at least) the third one in the suffix $v_m \rightarrow \cdots
  \rightarrow v_1 \rightarrow p$.  (We say that it is the third,
  rather than the first, because paths are constructed backwards).

  When a path $p'[v',k']$ reaches a node $v'$ for the third time, the
  mark of $v'$ for the keyword $k'$ has already been changed to $\inan$
  in a previous iteration. This follows from the following two
  observations. First, the first path to reach a node $v'$ is also the
  one to change its mark to $\vis$. Second, a path that reaches a node
  marked as $\vis$ can be unfrozen only when that mark is changed to $\inan$.

  Let $v_m \rightarrow \cdots \rightarrow v_1$ be a repeated run in
  $\bar p$ and suppose that $m>n=|V|$. Hence, there is a node $v_i$
  that appears twice in the repeated run; that is, there is a $j<i$,
  such that $v_j=v_i$.  If the path $v_i \rightarrow \cdots
  \rightarrow v_1 \rightarrow p$ is considered in the loop of
  line~\ref{alg:gtf:relaxStart}, then it would fail the test of
  line~\ref{alg:gtf:testEssential} (because, as explained earlier, all
  the nodes on the cycle $v_i \rightarrow \cdots \rightarrow v_j$ are already
  marked as $\inan$).  We conclude that the algorithm does not
  construct paths that have a repeated run with more than $n$ nodes.

  It thus follows that two disjoint repeated runs of a constructed
  path $\bar p$ must be separated by a node that appears (in a
  position between them) for the first or second time. A path can have
  at most $2n$ positions, such that in each one a node appears for the
  first or second time. Therefore, if a path $\bar p$ is constructed
  by the algorithm, then it can have at most $2n(n+1)$ nodes.
  Using $n$ distinct nodes, we can construct at most
  $(n+1)^{2n(n+1)}$ paths with $2n(n+1)$ or fewer nodes.

\end{proof}

\begin{lemma}\label{lemma:all-roots}
$K$-Roots have the following two properties.
\begin{enumerate}
\item\label{part:all-roots-part-one}
All the $K$-roots are discovered before the algorithm terminates.
Moreover, they are discovered in the increasing order of their best heights.
\item\label{part:all-roots-part-two}
Suppose that $r$ is a $K$-root with a best height $b$.
If $p[v,k]$ is a path (from any node $v$ to any keyword $k$)
that is exposed before the iteration that discovers $r$ as a $K$-root,
then $\fweight(p[v,k]) \le b$.
\end{enumerate}
\end{lemma}
\begin{proof}
We first prove Part~\ref{part:all-roots-part-one}. 
Suppose that a keyword $k$ is reachable from node $v$.
As long as $v.\mathit{marks[k]}$ is $\act$ at the beginning
of the main loop (line~\ref{alg:gtf:mainLoop_start}), 
Lemma~\ref{LEMMA:ON-QUEUE} implies that 
the queue $Q$ contains (at least) one suffix of a minimal path from $v$ to $k$.
By Lemma~\ref{lemma:finite}, the algorithm constructs a finite number of
paths.  By Proposition~\ref{prop:twice}, the same path can be inserted
into the queue at most twice. 
Since the algorithm does not terminate while $Q$ is not empty,
$v.\mathit{marks[k]}$ must be changed to $\vis$ after a finite time.
It thus follows that each $K$-root is discovered after a finite time.

Next, we show that the $K$-roots are discovered in the increasing order of
their best heights. Let $r_1$ and $r_2$ be two $K$-roots with best heights
$b_1$ and $b_2$, respectively, such that $b_1 < b_2$.
Lemma~\ref{LEMMA:GTF-SHORTEST-PATH-MARKS} implies the following for
$r_i$ ($i=1,2$).  For all keywords $k\in K$, a minimal path from $r_i$
to $k$ changes $r_i.\mathit{marks[k]}$ from $\act$ to $\vis$; that is,
$r_i$ is discovered as a $K$-root by minimal paths.
Suppose, by way of contradiction, that $r_2$ is discovered first.
Hence, a path with weight $b_2$ is removed from $Q$ while
Lemma~\ref{LEMMA:ON-QUEUE} implies that
a suffix with a weight of at most $b_1$ is still on $Q$.
This contradiction completes the proof of Part~\ref{part:all-roots-part-one}. 

Now, we prove Part~\ref{part:all-roots-part-two}. 
As shown in the proof of Part~\ref{part:all-roots-part-one},
a $K$-root is discovered by minimal paths. 
Let $r$ be a $K$-root with best height $b$.
Suppose, by way of contradiction, that a path $p[v,k]$,
such that $\fweight(p[v,k]) > b$,
is exposed before the iteration, say $i$, that discovers $r$ as a $K$-root.
By Lemma~\ref{LEMMA:ON-QUEUE}, at the beginning of iteration $i$,
the queue $Q$ contains a suffix with weight of at most $b$. Hence,
$p[v,k]$ cannot be removed from $Q$ at the beginning of iteration $i$.
This contradiction proves Part~\ref{part:all-roots-part-two}. 
\end{proof}

\begin{restatable}{lemma}{lemmaIH}
\label{LEMMA:INCREASING-HEIGHT}
Suppose that node $v$ is discovered as a $K$-root at iteration $i$.
Let $p_1[v',k']$ and $p_2[v,k]$ be paths that are exposed in iterations 
$j_1$ and $j_2$, respectively. If $i < j_1 < j_2$, then 
$\fweight(p_1[v',k']) \le \fweight(p_2[v,k])$.
Note that $k$ and $k'$ are not necessarily the same and similarly for
$v$ and $v'$; moreover, $v'$ has not necessarily been discovered as a $K$-root. 
\end{restatable}

\begin{proof}
Suppose the lemma is false. In particular, consider an iteration $j_1$
of the main loop (line~\ref{alg:gtf:mainLoop_start})
that violates the lemma. That is, the following hold in iteration $j_1$.
\begin{itemize}
\item Node $v$ has already been discovered as a $K$-root in an earlier
  iteration (so, there are no frozen paths at $v$).
\item A path $p_1[v',k']$ is exposed in iteration $j_1$.
\item A path $p_2[v,k]$ having a strictly lower weight than $p_1[v',k']$ 
(i.e.,~$\fweight(p_2[v,k]) < \fweight(p_1[v',k'])$) will be
  exposed after iteration $j_1$. Hence, a proper suffix of this path is
  frozen at some node $z$ during iteration $j_1$.
\end{itemize}

For a given $v$ and $p_1[v',k']$, there could be several paths
$p_2[v,k]$ that satisfy the third condition above. We choose one, such
that its suffix is frozen at a node $z$ that is closest from $v$.
Since $v$ has already been discovered as a $K$-root,
$z$ is different from $v$.

Clearly, $z.\mathit{marks[k]}$ is changed to $\vis$ before iteration
$j_1$.  By Lemma~\ref{LEMMA:GTF-SHORTEST-PATH-MARKS}, a minimal path
$p_m[z,k]$ does that.  Let $p_s[v,z]$ be a minimal path from $v$ to
$z$.

Consider the path 
\begin{equation*}
\bar p[v,k] = p_s[v,z] \concat p_m[z,k].
\end{equation*}  
Since both $p_s[v,z]$ and $p_m[z,k]$ are minimal, the weight
of their concatenation (i.e.,~$\bar p[v,k]$) is no more than that of
$p_2[v,k]$ (which is also a path that passes through node $z$).  Hence,
$\fweight(\bar p[v,k]) < \fweight(p_1[v',k'])$.

We now show that the conditions of Lemma~\ref{LEMMA:PATH-CONCAT} are 
satisfied at the beginning of iteration $j_1$
(i.e.,~$j_1$ corresponds to $i$ in Lemma~\ref{LEMMA:PATH-CONCAT}).
Conditions~\ref{cond:new-first}--\ref{cond:new-second} clearly hold.
Condition~\ref{cond:first} is satisfied because a suffix of $p_2[v,k]$
is frozen at $z$.
Condition~\ref{cond:second} holds, because
of the choice of $z$ and the fact
$\fweight(\bar p[v,k]) < \fweight(p_1[v',k'])$ that was shown earlier.
Condition~\ref{cond:third} holds, because otherwise
$\bar p[v,k]$ would be unfrozen and $z.\mathit{marks[k]}$ would be $\inan$
rather than $\vis$.

By Lemma~\ref{LEMMA:PATH-CONCAT}, a suffix of $\bar p[v,k]$ is on the
queue at the beginning of iteration $j_1$.  This contradicts the
assumption that the path $p_1[v',k']$ is removed from the
queue at the beginning of iteration $j_1$, because $\bar p[v,k]$ (and, hence,
any of its suffixes) has a strictly lower weight.
\end{proof}

\begin{restatable}{lemma}{lemmaNVWT}
\label{LEMMA:NO-VISIT-WHEN-TERMINATING}
For all nodes $v\in \V$, such that $v$ is a $K$-root,
the following holds.
If $z$ is a node on a simple path from $v$ to some $k\in K$,
then $z.\mathit{marks[k]} \not= \vis$ when the algorithm terminates. 
\end{restatable}

\begin{proof} 
The algorithm terminates when the test of
line~\ref{alg:gtf:mainLoop_start} shows that $Q$ is empty.  Suppose
that the lemma is not true.  Consider some specific $K$-root $v$ and
keyword $k$ for which the lemma does not hold. Among all the nodes $z$
that violate the lemma with respect to $v$ and $k$, let $z$ be a
closest one from $v$.  Observe that $z$ cannot be $v$, because of the
following two reasons.  First, by Lemma~\ref{lemma:all-roots}, node $v$ is
discovered as a $K$-root before termination. Second, when a $K$-root is
discovered (in
lines~\ref{alg:gtf:becomesRoot}--\ref{alg:gtf:isRootGetsTrue}), all
its marks become $\inan$ in
lines~\ref{alg:gtf:newRootStart}--\ref{alg:gtf:unfrNewRoot}.

Suppose that $p_m[z,k]$ is the path that changes $z.\mathit{marks[k]}$
to $\vis$.  Let $p_s[v,z]$ be a minimal path from $v$ to $z$. Note
that $p_s[v,z]$ exists, because $z$ is on a simple path from $v$ to $k$.
Consider the path 
\begin{equation*}
\bar p[v,k] =  p_s[v,z] \concat p_m[z,k].  
\end{equation*}

Suppose that the test of line~\ref{alg:gtf:mainLoop_start} is \textbf{false}
(and, hence, the algorithm terminates) on iteration $i$.
We now show that the conditions of Lemma~\ref{LEMMA:PATH-CONCAT}
are satisfied at the beginning of that iteration.
Conditions~\ref{cond:new-first}--\ref{cond:new-second} of
Lemma~\ref{LEMMA:PATH-CONCAT} clearly hold. 
Conditions~\ref{cond:first}--\ref{cond:second} are satisfied
because of how $z$ is chosen. Condition~\ref{cond:third} holds, because
otherwise $z.\mathit{marks[k]}$ should have been changed to $\inan$.

By Lemma~\ref{LEMMA:PATH-CONCAT}, a suffix of $\bar p[v,k]$ is on $Q$
when iteration $i$ begins, contradicting our assumption that $Q$ is empty.
\end{proof}

\begin{theorem}\label{theorem:gtf-correct}
GTF is correct. In particular, it finds all and only answers
to the query $K$ by increasing height within $2(n+1)^{2n(n+1)}$ iterations
of the main loop (line~\ref{alg:gtf:mainLoop_start}), where $n=|V|$. 
\end{theorem}
\begin{proof}
By Lemma~\ref{lemma:finite},
the algorithm constructs at most $(n+1)^{2n(n+1)}$ paths.
By Proposition~\ref{prop:twice},
a path can be inserted into the queue $Q$ at most twice.
Thus, the algorithm terminates after at most $2(n+1)^{2n(n+1)}$ iterations
of the main loop.

By Part~\ref{part:all-roots-part-one} of Lemma~\ref{lemma:all-roots},
all the $K$-roots are discovered.
By Lemma~\ref{LEMMA:NO-VISIT-WHEN-TERMINATING},
no suffix of a simple path from a $K$-root to a keyword can be frozen
upon termination. Clearly, no such suffix can be on $Q$
when the algorithms terminates. Hence, the algorithm constructs 
all the simple paths from each $K$-root to every 
keyword. It thus follows that the algorithm finds all the answers to $K$.
Clearly, the algorithm generates only valid answers to $K$.

Next, we prove that the answers are produced in the order of
increasing height.  So, consider answers $a_1$ and $a_2$ 
that are produced in iterations $j'_1$ and $j_2$, respectively.  
For the answer $a_i$ ($i=1,2$),
let $r_i$ and $h_i$ be its $K$-root and height, respectively.
In addition, let $b_i$ be the best height of $r_i$ ($i=1,2$).

Suppose that $j'_1 < j_2$. We have to prove that $h_1 \le h_2$.
By way of contradiction, we assume that $h_1 > h_2$.
By the definition of best height, $h_2 \ge b_2$.  Hence, $h_1 > b_2$.

Let $p_2[r_2,k]$ be the path of $a_2$ that is exposed 
(i.e.,~removed from $Q$) in iterations $j_2$.
Suppose that $p_1[r_1,k']$ is a path of $a_1$, such that 
$\fweight(p_1[r_1,k'])=h_1$ and $p_1[r_1,k']$ is exposed in the iteration $j_1$
that is as close to iteration $j'_1$ as possible
(among all the paths of $a_1$ from $r_1$ to a keyword with a weight
equal to $h_1$). Clearly, $j_1 \le j'_1$ and hence $j_1 < j_2$.

We now show that $\fweight(p_1[r_1,k']) < h_1$, in contradiction to
$\fweight(p_1[r_1,k'])=h_1$. Hence, the claim that $h_1 \le h_2$ follows.
Let $i$ be the iteration that discovers $r_2$ as a $K$-root.
There are two cases to consider as follows.

\begin{description}
\item[Case 1: $i < j_1$.] In this case, $i < j_1 < j_2$, since $j_1 <
  j_2$.  By Lemma~\ref{LEMMA:INCREASING-HEIGHT},
  $\fweight(p_1[r_1,k']) \le \fweight(p_2[r_2,k])$.  (Note that we
  apply Lemma~\ref{LEMMA:INCREASING-HEIGHT} after replacing $v$ and
  $v'$ with $r_2$ and $r_1$, respectively.)  Hence,
  $\fweight(p_1[r_1,k']) < h_1$, because $\fweight(p_2[r_2,k]) \le h_2$
  follows from the definition of height and we have assumed that $h_1 > h_2$.
\item[Case 2: $j_1 \le i$.]
  By Part~\ref{part:all-roots-part-two} of Lemma~\ref{lemma:all-roots},
  $\fweight(p_1[r_1,k']) \le b_2$. Hence, $\fweight(p_1[r_1,k']) < h_1$,
  because we have shown earlier that $h_1 > b_2$.
\end{description}
Thus, we have derived a contradiction and, hence,
it follows that answers are produced by increasing height.
\end{proof}

\begin{corollary}\label{cor:running}
  The running time of the algorithm GTF is
  $O\left(kn(n+1)^{2kn(n+1)+1}\right)$, where $n$ and $k$ are the number of
  nodes in the graph and keywords in the query, respectively.
\end{corollary}
\begin{proof}
  The most expensive operation is a call to $\prodAns(v.paths,p)$. By
  Lemma~\ref{lemma:finite}, there are at most $(n+1)^{2n(n+1)}$ paths.
  A call to the procedure $\prodAns(v.paths,p)$ considers all combinations of
  $k-1$ paths plus $p$. For each combination, all its $k$ paths are
  traversed in linear time. Thus, the total cost of one call to
  $\prodAns(v.paths,p)$ is $O\left(kn(n+1)(n+1)^{(k-1)2n(n+1)}\right)$.  By
  Theorem~\ref{theorem:gtf-correct}, there are at most
  $2(n+1)^{2n(n+1)}$ iterations. Hence, the running time is
  $O\left(kn(n+1)^{2kn(n+1)+1}\right)$.
\end{proof}

\section{Summary of the Experiments}\label{sec:experiments_sum}
In this section, we summarize our experiments.
The full description of the methodology and results is
given in Appendix~\ref{sec:experiments}
.
We performed extensive experiments to measure the efficiency of GTF.
The experiments were done on the
Mondial\footnote{http://www.dbis.informatik.uni-goettingen.de/Mondial/}
and DBLP\footnote{http://dblp.uni-trier.de/xml/} datasets.

To test the effect of freezing, we ran the naive approach (described
in Section~\ref{sec:naive}) and GTF on both datasets.  We measured the
running times of both algorithms for generating the top-$k$ answers
($k=100, 300, 1000$).  We discovered that the freezing technique gives
an improvement of up to about one order of magnitude.  It has a greater
effect on Mondial than on DBLP, because the former is highly cyclic
and, therefore, has more paths (on average) between a pair of nodes.
Freezing has a greater effect on long queries than short ones. This is
good, because the bigger the query, the longer it takes to produce its
answers.  This phenomenon is due to the fact that the average height
of answers increases with the number of keywords. Hence, the naive
approach has to construct longer (and probably more) paths that do not
contribute to answers, whereas GTF avoids most of that work.

In addition, we compared the running times of GTF with those of
BANKS~\cite{icdeBHNCS02,vldbKPCSDK05}, BLINKS~\cite{sigmodHWYY07},
SPARK~\cite{tkdeLWLZWL11} and ParLMT~\cite{pvldbGKS11}.  The last one
is a parallel implementation of~\cite{sigmodGKS08}; we used its
variant ES (early freezing with single popping) with 8 threads.  BANKS
has two versions, namely, MI-BkS~\cite{icdeBHNCS02} and
BiS~\cite{vldbKPCSDK05}. The latter is faster than the former by up to one
order of magnitude and we used it for the running-time comparison.

GTF is almost always the best, except in two particular cases.
First, when generating $1,000$ answers over Mondial, SPARK is
better than GTF by a tiny margin on queries with 9 keywords, but is
slower by a factor of two when averaging over all queries.  On DBLP, however,
SPARK is slower than GTF by up to two orders of magnitude.  Second, when
generating $100$ answers over DBLP, BiS is slightly better than GTF on
queries with 9 keywords, but is 3.5 times slower when averaging over all
queries.  On Mondial, however, BiS is slower than GTF by up to one
order of magnitude.  All in all, BiS is the second best algorithm in most
of the cases. The other systems are slower than GTF by one to two
orders of magnitude.

Not only is our system faster, it is also
increasingly more efficient as either the number of generated answers
or the size of the data graph grows.  This may seem counterintuitive,
because our algorithm is capable of generating all paths (between a
node and a keyword) rather than just the minimal one(s).  However, our
algorithm generates non-minimal paths only when they can potentially
contribute to an answer, so it does not waste time on doing useless
work.  Moreover, if only minimal paths are constructed, then longer
ones may be needed in order to produce the same number of answers,
thereby causing more work compared with an algorithm that is capable
of generating all paths.

GTF does not miss answers (i.e.,~it is capable of generating all of them).
Among the other systems we tested, 
ParLMT~\cite{pvldbGKS11} has this property and is theoretically
superior to GTF, because it enumerates
answers with polynomial delay (in a 2-approximate order of increasing height),
whereas the delay of GTF could be exponential.  In our experiments,
however, ParLMT was slower by two orders of magnitude, even though it
is a parallel algorithm (that employed eight cores in our tests).
Moreover, on a large dataset, ParLMT ran out of memory when the query had seven 
keywords. The big practical advantage of GTF over ParLMT 
is explained as follows.
The former constructs paths incrementally whereas the latter
(which is based on the Lawler-Murty procedure~\cite{mansciL72,orM68}) 
has to solve a new optimization problem for each produced answer,
which is costly in terms of both time and space.

A critical question is how important it is to have an algorithm that
is capable of producing all the answers.  We compared our algorithm
with BANKS. Its two versions only generate answers consisting
of minimal paths and, moreover, those produced by BiS have distinct
roots.  BiS (which is overall the second most efficient system in our
experiments) misses between $81\%$ (on DBLP) to $95\%$ (on Mondial) of
the answers among the top-$100$ generated by GTF.  MI-BkS misses much
fewer answers, that is, between $1.8\%$ (on DBLP) and $32\%$ (on
Mondial), but it is slower than BiS by up to one order of magnitude. For
both versions the percentage of misses increases as the number of
generated answers grows.  This is a valid and significant comparison,
because our algorithm generates answers in the same order as BiS and
MI-BkS, namely, by increasing height.

\section{Conclusions}
We presented the GTF algorithm for enumerating, by increasing height,
answers to keyword search over data graphs.  Our main contribution is
the freezing technique for avoiding the construction of (most if not
all) non-minimal paths until it is determined that they can reach
$K$-roots (i.e.,~potentially be parts of answers).  Freezing is an
intuitive idea, but its incorporation in the GTF algorithm involves
subtle details and requires an intricate proof of correctness.  In
particular, cyclic paths must be constructed (see
Appendix~\ref{sec:looping_paths}
), although they are not part of any
answer.  For efficiency's sake, however, it is essential to limit the
creation of cyclic paths as much as possible, which is accomplished by
lines~\ref{alg:gtf:reassignFalse} and~\ref{alg:gtf:testEssential} of
Figure~\ref{alg:gtf}.

Freezing is not merely of theoretical importance.  Our extensive
experiments (described in Section~\ref{sec:experiments_sum}
and  Appendix~\ref{sec:experiments}
) show
that freezing increases efficiency by up to about one order of magnitude
compared with the naive approach (of Section~\ref{sec:naive}) that
does not use it.

The experiments of Section~\ref{sec:experiments_sum} and 
Appendix~\ref{sec:experiments} 
also show that in
comparison to other systems, GTF is almost always the best, sometimes
by several orders of magnitude.
Moreover, our algorithm is more scalable than other systems.
The efficiency of GTF is a significant achievement especially in light
of the fact that it is complete (i.e.,~does not miss answers). 
Our experiments show that some of the other systems 
sacrifice completeness for the sake of efficiency.
Practically, it means that they generate longer paths resulting
in answers that are likely to be less relevant than the missed ones.

The superiority of GTF over ParLMT is an indication that polynomial
delay might not be a good yard stick for measuring the practical
efficiency of an enumeration algorithm. An important topic for future
work is to develop theoretical tools that are more appropriate for
predicting the practical efficiency of those algorithms.

\bibliographystyle{plain}
\bibliography{main}

\newpage

\appendix
\section{The Need for Essential Paths }\label{sec:looping_paths}
In this section, we give an example showing that essential cyclic paths must
be constructed by the GTF algorithm, in order not to miss some answers.

Suppose that we modify the test of line~\ref{alg:gtf:testEssential}
of Figure~\ref{alg:gtf} to be ``$v'$ is not on $p$'' (i.e.,~we omit
the second part, namely, ``$v' \rightarrow p$ is essential'').
Note that the new test means that only acyclic paths are inserted into $Q$.

\tikzstyle{vertex}=[circle,draw,minimum size=14pt,inner sep=1pt]   
\tikzstyle{keyword}=[rectangle,draw,minimum size=14pt,inner sep=1pt]  
\usetikzlibrary{positioning}

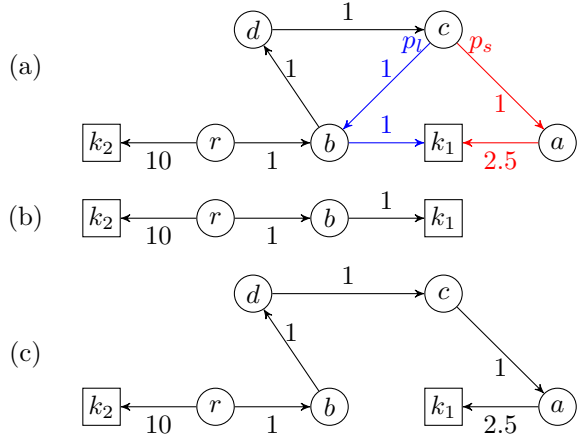
\begin{figure}[tbh]
\centering
\begin{tikzpicture}[->,>=stealth',shorten >=0pt,auto,node distance=1.5cm]

  
  \node[keyword] (k1) at (0,0) {$k_1$};
  \node[vertex] (a) [right of=k1]{$a$};
  \node[vertex] (b) [left of=k1]{$b$};  
  \node[vertex] (c) [above of=k1]{$c$};  
  \node[vertex] (d) [left = 2 of c]{$d$};  
  \node[vertex] (r) [left of=b]{$r$};  
  \node[keyword] (k2)[left of=r]  {$k_2$};
  \node[blue] (pl) at (-0.4,1.3) {$p_l$};
  \node[red] (ps) at (0.5,1.3) {$p_s$};
  
  \draw (a) edge[red] node {2.5} (k1);
  \draw (b) edge[blue] node {1} (k1);
  \draw (c) edge[blue] node[above] {1} (b);
  \draw (c) edge[red]node [below]{1} (a);
  \draw (d) edge node {1} (c);
  \draw (b) edge node [above] {1} (d);
  \draw (r) edge node [below] {1} (b);
  
  \draw (r) edge node {10} (k2);
  
  \node at (-5.5,1) {(a)};


  \node[keyword] (k1_a1) at (0,-1) {$k_1$};
  \node[vertex] (b_a1) [left of=k1_a1]{$b$};  
  \node[vertex] (r_a1) [left of=b_a1]{$r$};  
  \node[keyword] (k2_a1)[left of=r_a1]  {$k_2$};
    
  \draw (b_a1) edge node {1} (k1_a1);
  \draw (r_a1) edge node [below] {1} (b_a1);
  
  \draw (r_a1) edge node {10} (k2_a1);

  \node at (-5.5,-1) {(b)};
  \node[keyword] (k1_a2) at (0,-3.5) {$k_1$};
  \node[vertex] (a_a2) [right of=k1_a2]{$a$};
  \node[vertex] (b_a2) [left of=k1_a2]{$b$};  
  \node[vertex] (c_a2) [above  of=k1_a2]{$c$};  
  \node[vertex] (d_a2) [left = 2 of c_a2]{$d$};  
  \node[vertex] (r_a2) [left of=b_a2]{$r$};  
  \node[keyword] (k2_a2)[left of=r_a2]  {$k_2$};
    
  \draw (a_a2) edge node {2.5} (k1_a2);
  \draw (c_a2) edge node [below]{1} (a_a2);
  \draw (d_a2) edge node {1} (c_a2);
  \draw (b_a2) edge node [above] {1} (d_a2);
  \draw (r_a2) edge node [below] {1} (b_a2);
  
  \draw (r_a2) edge node {10} (k2_a2);
  \node at (-5.5,-2.8) {(c)};
\end{tikzpicture}
\caption{\label{fig:freezing_looping_pahts} 
(a)~The graph snippet (b) The first answer (c) The second answer}
\end{figure} 

Consider the data graph of Figure~\ref{fig:freezing_looping_pahts}(a).
The weights of the nodes are 1 and the weight of each edge appears next to it.
Suppose that the query is $\{k_1,k_2\}$.
We now show that the modified GTF algorithm would miss 
the answer presented in Figure~\ref{fig:freezing_looping_pahts}(c),
where the root is $r$.

The initialization step of the algorithm inserts into $Q$ 
the paths $k_1$ and $k_2$, each consisting of a single keyword.  
Next, we describe the iterations of the main loop.
For each one, we specify the path that is removed from $Q$
and the paths that are inserted into $Q$.
We do not mention explicitly how the marks are changed,
unless it eventually causes freezing.
\begin{enumerate}
\item
The path $k_1$ is removed from $Q$, and
the paths $a \rightarrow k_1$ and $b \rightarrow k_1$ are inserted into $Q$.
\item
The path $k_2$ is removed from $Q$ and
the path $r \rightarrow k_2$ is inserted into $Q$.  
\item
The path $b \rightarrow k_1$ is removed (since its weight is the
lowest on $Q$), and $r \rightarrow b \rightarrow k_1$ and
$p_l = c \rightarrow b \rightarrow k_1$ are inserted
(the latter is shown in blue in Figure~\ref{fig:freezing_looping_pahts}(a)).
\item
The path $a \rightarrow k_1$ is removed from $Q$ and the path
$p_s= c \rightarrow a \rightarrow k_1$ (shown in red) is inserted into $Q$.
\item
The path $p_l$ is removed from $Q$ and
$c.\mathit{marks}[k_1]$ is changed to $\vis$; then the path $d \rightarrow p_l$
is inserted to $Q$.
\item
The path $r \rightarrow b \rightarrow k_1$ is
removed from $Q$ and nothing is inserted into $Q$.
\item
The $p_s$ is removed from $Q$ and freezes at node $c$.
\item
The path $d \rightarrow p_l$ is removed from $Q$.
It can only be extended further to node $b$, but that would create a cycle,
so nothing is inserted into $Q$.
\end{enumerate}
Eventually, node $r$ is discovered to be a $K$-root. However,
$c.\mathit{marks}[k_1]$ will never be changed from $\vis$ to $\inan$
for the following reason.  
The minimal path that first visited $c$ (namely, $p_l$) must make a
cycle to reach $r$. Therefore, the path $p_s$ remains frozen at
node $c$ and the answer of Figure~\ref{fig:freezing_looping_pahts}(c)
will not be produced.
\section{Experiments}\label{sec:experiments}
In this appendix, we discuss the methodology and results of the experiments.
First, we start with a description of the datasets and the queries.
Then, we discuss the setup of the experiments.
Finally, we present the results.
This appendix expands the summary given in Section~\ref{sec:experiments_sum}.

\subsection{Datasets and Queries\label{sec:datasets}}
Most of the systems we tested are capable of parsing a relational
database to produce a data graph. Therefore, we selected this
approach.  The experiments were done on
Mondial\footnote{http://www.dbis.informatik.uni-goettingen.de/Mondial/}
and DBLP\footnote{http://dblp.uni-trier.de/xml/} that are commonly
used for testing systems for keyword-search over data graphs.
Mondial is a highly connected graph.
We specified 65 foreign keys in Mondial,
because the downloaded version lacks such definitions.

DBLP has only an XML version that is available for download.  In
addition, that version is just a bibliographical list.  In order to
make it more meaningful, we modified it as follows.  We replaced
the citation and cross-reference elements with IDREFS.  Also, we
created a unique element for each \emph{author, editor} and
\emph{publisher}, and replaced each of their original occurrences with
an IDREF.  After doing that, we transformed the XML dataset into a
relational database.  To test scalability, we also created a subset of DBLP.
The two versions of DBLP are called \emph{Full DBLP} and \emph{Partial DBLP}.

Table~\ref{tbl:dataset_stats} gives the sizes of the data graphs.  The
numbers of nodes and edges exclude the keyword nodes and their
incident edges.  The average degree of the Mondial
data graph is the highest, due to its high connectivity.

We manually created queries for Mondial and DBLP.
The query size varies from 2 to 10 keywords, and there
are four queries of each size.

\begin{table}[t]
\centering
\caption{\label{tbl:dataset_stats}The sizes of the data graphs}
\begin{tabular}{l | c | c | c |}

Dataset Name  &  nodes &  edges & average degree\\ \hline
Mondial & 21K 	& 86K 	& 4.04\\
Partial DBLP   &  649K & 1.6M 	& 2.48\\
Full DBLP    & 7.5M 	& 21M 	& 2.77\\
\hline
\end{tabular}
\end{table}

In principle, answers should be ranked according to their weight.
However, answers cannot be generated efficiently by increasing weight.
Therefore, a more tractable order (e.g.,~by increasing height) is used
for generating answers, followed by sorting according to the desired ranking. 
We employed the weight function described in\cite{pvldbGKS11},
because based on our experience it is effective in practice.
The other tested systems have their own weight functions. 

We have given the reference to the weight function 
so that the presentation of our work is complete.
However, this paper is about efficiency of search algorithms, rather than their effectiveness.
Therefore, measuring recall and precision is beyond the scope of this paper.

\subsection{The Setup of the Experiments}\label{sec:setup}

We ran the tests on a server with two quad-core, 2.67GHz Xeon X5550
processors and 48GB of RAM. We used Linux Debian (kernel 3.2.0-4-amd64),
Java 1.7.0\_60, PostgreSQL 9.1.13 and MySQL 5.5.35.  
The Java heap was allocated 10GB of RAM.
Systems that use a DBMS got additional 5GB of RAM for the buffer (of the DBMS).

The tests were executed on BANKS~\cite{icdeBHNCS02,vldbKPCSDK05},
BLINKS~\cite{sigmodHWYY07}, SPARK~\cite{tkdeLWLZWL11}, GTF (our
algorithm) and ParLMT~\cite{pvldbGKS11}.  The latter is a parallel
implementation of~\cite{sigmodGKS08}; we used its
variant ES (early freezing with single popping) with 8 threads.
BANKS has two versions, namely, MI-BkS~\cite{icdeBHNCS02} and
BiS~\cite{vldbKPCSDK05}. The latter is faster than the former by one
order of magnitude and we used it for the running-time comparison.
We selected these systems because their code is available.

Testing the efficiency of the systems is a bit like comparing apples
and oranges.  They use different weight functions and produce answers
in dissimilar orders.  Moreover, not all of them are capable of
generating all the answers.  We have striven to make the tests as equitable
as possible.
The exact configuration we used for each system, as well as some other details,
are given in Appendix~\ref{sec:config_details}.

We measured the running times as a function of the query size.
The graphs show the average over the four queries of each size.
The time axis is logarithmic, because there are order-of-magnitude differences
between the fastest and slowest systems.

\subsection{The Effect of Freezing}

To test the effect of freezing, we ran the naive approach (described in Section~\ref{sec:naive})
and GTF on all the 
datasets.  Usually, a query has thousands of answers,
  because of multiple paths between pairs of nodes.  We executed the
  algorithms until the first $100$, $300$ and $1,000$ answers were
  produced.  
  Then, we measured the speedup of GTF over the naive approach.
  Typical results are shown in
Figure~\ref{fig:GTFSpeedup}.
Note that smaller queries are not necessarily subsets of
  larger ones. Hence, long queries may be computed faster
  than short ones. Generally, the computation time depends on the
  length of the query and the percentage of nodes in which the
  keywords appear.

\pgfplotsset{every axis legend/.append style={at={(0.4,0.85)},anchor=south}}
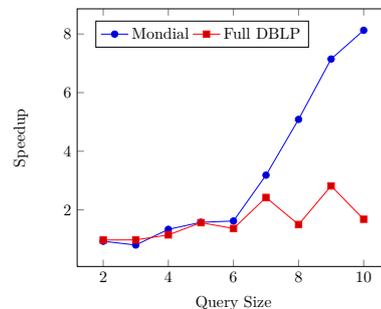
\begin{wrapfigure}{r}{0.42\textwidth}
\vspace{-20pt}
\begin{center}
  \begin{tikzpicture} [scale=0.6]
  \begin{axis}[xmode=normal,ymode=normal,xlabel={Query Size},{ylabel=Speedup},legend columns=4]
  \addplot table[x=Size,y=M_120]{datafiles/speedup.data};
  \addplot table[x=Size,y=FD_120]{datafiles/speedup.data};
  \legend{Mondial,
   Full DBLP}
  \end{axis}
  \end{tikzpicture}
\vskip-1em
\caption{\label{fig:GTFSpeedup}The speedup of GTF over the naive approach (for $100$ answers)}
\end{center}
\end{wrapfigure} 

The freezing technique gives an improvement of up to about one order of
magnitude.  It has a greater effect on Mondial than on DBLP, because
the former is highly cyclic and, therefore, there are more paths (on
average) between a pair of nodes. The naive approach produces all those paths,
whereas GTF freezes the construction of most of them until they are needed.

On all datasets, freezing has a greater effect on long queries than
short ones. This is good, because the bigger the query, the longer it
takes to produce its answers.  This phenomenon is due to the fact that the
average height of answers increases with the number of
keywords. Hence, the naive approach has to construct longer (and
probably more) paths that do not contribute to answers, whereas GTF
avoids most of that work.  In a few cases, GTF is slightly slower than
the naive approach, due to its overhead.
However, this happens for short queries that are computed very fast anyhow;
so, this is not a problem.

\pgfplotsset{every axis legend/.append style={at={(0.35,0.8)},anchor=south}}
\begin{figure*}[b]
\begin{minipage}{0.33\textwidth}
\begin{center}
  \begin{tikzpicture} [scale=0.55]
  \begin{axis}[xmode=normal,ymode=log,xlabel={Query Size},{ylabel=Time (sec)},legend columns=2]
  \addplot [color=blue!50!black, densely dotted, every mark/.append style={solid, fill=blue}, mark=*]
  		   table[x=Size,y=M_100,y error=M_100_DEV]{datafiles/BANKS.data};
  \addplot [color=red!50!black,  every mark/.append style={solid, fill=red}, mark=square*]
  		   table[x=Size,y=M_120,y error=M_120_DEV]{datafiles/GTF.data};
  \addplot [color=brown!50!black, dashdotted, every mark/.append style={solid, fill=brown}, mark=triangle*]
  		   table[x=Size,y=M_100,y error=M_100_DEV]{datafiles/SPARK.data};
  \addplot table[x=Size,y=M_120]{datafiles/ParLaw.data};
  \legend{BANKS, GTF, SPARK, ParLMT}
  \end{axis}
  \end{tikzpicture}  
\vskip-1em
\caption{\label{fig:runtimeMondial100}Mondial 100 ans.}
\end{center}
\end{minipage}
\hfill	
\begin{minipage}{0.32\textwidth}
\begin{center}
  \begin{tikzpicture} [scale=0.55]
  \begin{axis}[xmode=normal,ymode=log,xlabel={Query Size},legend columns=2]
  \addplot [color=blue!50!black, densely dotted, every mark/.append style={solid, fill=blue}, mark=*]
  		   table[x=Size,y=M_300,y error=M_300_DEV]{datafiles/BANKS.data};
  \addplot [color=red!50!black,  every mark/.append style={solid, fill=red}, mark=square*]
           table[x=Size,y=M_320,y error=M_320_DEV]{datafiles/GTF.data};
  \addplot [color=brown!50!black, dashdotted, every mark/.append style={solid, fill=brown}, mark=triangle*]
           table[x=Size,y=M_300,y error=M_300_DEV]{datafiles/SPARK.data};
  \addplot table[x=Size,y=M_320]{datafiles/ParLaw.data};
  \end{axis}
  \end{tikzpicture}  
\vskip-1em
\caption{\label{fig:runtimeMondial300}Mondial 300 ans.}
\end{center}
\end{minipage}
\hfill
\begin{minipage}{0.32\textwidth}
\begin{center}
  \begin{tikzpicture} [scale=0.55]
  \begin{axis}[xmode=normal,ymode=log,xlabel={Query Size},legend columns=4]
  \addplot [color=blue!50!black, densely dotted, every mark/.append style={solid, fill=blue}, mark=*]
  			table[x=Size,y=M_1000,y error=M_1000_DEV]{datafiles/BANKS.data};
  \addplot [color=red!50!black,  every mark/.append style={solid, fill=red}, mark=square*]
  			table[x=Size,y=M_1020,y error=M_1020_DEV]{datafiles/GTF.data};
  \addplot [color=brown!50!black, dashdotted, every mark/.append style={solid, fill=brown}, mark=triangle*]
  			table[x=Size,y=M_1000,y error=M_1000_DEV]{datafiles/SPARK.data};
  \addplot table[x=Size,y=M_1020]{datafiles/ParLaw.data};
  \end{axis}
  \end{tikzpicture}  
\vskip-1em
\caption{\label{fig:runtimeMondial1000}Mondial 1000 ans.}
\end{center}
\end{minipage}
\end{figure*}
 
\subsection{Comparison with Other Systems}\label{sec:comparison}

In this section, we compare the running times of the systems we
tested.  We ran each system to produce $100$, $300$ and $1,000$ answers
for all the queries and datasets.
Typical results are shown in
Figs~\ref{fig:runtimeMondial100}--\ref{fig:runtimeDblp100}.

GTF is almost always the best, sometimes by several orders of magnitude, 
except for a few cases in which it is second by a tiny margin.  

When running SPARK, many answers contained only some of the keywords, 
even though we set it up for the AND semantics (to be the same as the other
systems).

On Mondial, SPARK is generally the second best. 
Interestingly, its running time increases only slightly as the number of
answers grows. In particular, Figure~\ref{fig:runtimeMondial1000} shows that for
$1,000$ answers, SPARK is the second best by a wide margin, compared with 
BANKS which is third.
The reason for that is the way SPARK works. It produces expressions
(called \emph{candidate networks}) that are evaluated over a database.
Many of those expression are \emph{empty}, that is, produce no answer at all.
Mondial is a small dataset, so the time to compute a non-empty expression
is not much longer than computing an empty one. Hence, the running time hardly
depends on the number of generated answers.

However, on the partial and full versions of DBLP (which are large
datasets), SPARK is much worse than BANKS, even when producing only
$100$ answers (Figures~\ref{fig:runtimePdblp100}
and~\ref{fig:runtimeDblp100}).  We conclude that by and large BANKS
is the second best.

In Figures~\ref{fig:runtimeMondial100}--\ref{fig:runtimeMondial1000},
the advantage of GTF over BANKS grows as more answers are produced.
Figures~\ref{fig:runtimePdblp100} and~\ref{fig:runtimeDblp100} show
that the advantage of GTF increases as the dataset becomes larger.  We
conclude that GTF is more scalable than other systems when either the
size of the dataset or the number of produced answers is increased.

The running time of ParLMT grows linearly with the number of answers,
causing the gap with the other systems to get larger (as more answers
are produced).  This is because ParLMT solves a new optimization
problem from scratch for each answer.  In comparison, BANKS and GTF
construct paths incrementally, rather than starting all over again at
the keyword nodes for each answer.

ParLMT got an out-of-memory exception when computing long queries on
Full DBLP (hence, Figure~\ref{fig:runtimeDblp100} shows the
running time of ParLMT only for queries with at most six keywords).
This means that ParLMT is memory inefficient compared with the other systems.

BLINKS is the slowest system.  On Mondial, the average running time to
produce 100 answers is 46 seconds.  Since this is much worse than the
other systems,
Figures~\ref{fig:runtimeMondial100}--\ref{fig:runtimeMondial1000} do not
include  the graph for BLINKS.  On Partial DBLP, 
BLINKS is still the slowest (see Figure~\ref{fig:runtimePdblp100}).
On Full DBLP, BLINKS got an out-of-memory exception
during the construction of the indexes.
This is not surprising, because they store their indexes solely in main memory.

\pgfplotsset{every axis legend/.append style={at={(0.5,1.03)},anchor=south}}

\begin{figure}[b]
\begin{center}
\begin{minipage}{0.45\textwidth}
  \begin{tikzpicture} [scale=0.6]
  \begin{axis}[xmode=normal,ymode=log,xlabel={Query Size},{ylabel=Time (sec)},legend columns=5]
  \addplot [color=blue!50!black, densely dotted, every mark/.append style={solid, fill=blue}, mark=*]
  			table[x=Size,y=PD_100,y error=PD_100_DEV]{datafiles/BANKS.data};
  \addplot [color=red!50!black,  every mark/.append style={solid, fill=red}, mark=square*]
  			table[x=Size,y=PD_120,y error=PD_120_DEV]{datafiles/GTF.data};
  \addplot [color=brown!50!black, dashdotted, every mark/.append style={solid, fill=brown}, mark=triangle*]
  			table[x=Size,y=PD_100]{datafiles/SPARK.data};
  \addplot table[x=Size,y=PD_120]{datafiles/ParLaw.data};
  \addplot [color=green!30!black, dashed, every mark/.append style={solid, fill=green!50!black}, mark=diamond*]
  		   table[x=Size,y=PD_100]{datafiles/BLINKS.data};
  \legend{BANKS,GTF,SPARK,ParLMT, BLINKS}
  \end{axis}
  \end{tikzpicture}  
\caption{\label{fig:runtimePdblp100}Partial DBLP 100 answers}
\end{minipage}
\hfill
\begin{minipage}{0.45\textwidth}
  \begin{tikzpicture} [scale=0.6]
  \begin{axis}[xmode=normal,ymode=log,xlabel={Query Size},{ylabel=Time (sec)},legend columns=4]
  \addplot [color=blue!50!black, densely dotted, every mark/.append style={solid, fill=blue}, mark=*]
  			table[x=Size,y=FD_100,y error=FD_100_DEV]{datafiles/BANKS.data};
  \addplot [color=red!50!black,  every mark/.append style={solid, fill=red}, mark=square*]
  			table[x=Size,y=FD_120,y error=FD_120_DEV]{datafiles/GTF.data};
  \addplot [color=brown!50!black, dashdotted, every mark/.append style={solid, fill=brown}, mark=triangle*]
  		   table[x=Size,y=FD_100]{datafiles/SPARK.data};
  \addplot table[x=Size,y=FD_120]{datafiles/ParLaw.data};
  \legend{BANKS, GTF, SPARK, ParLMT}
  \end{axis}
  \end{tikzpicture}  
\caption{\label{fig:runtimeDblp100}Full DBLP 100 answers}
\end{minipage}
\end{center}
\end{figure}
\begin{table}[t]
\centering
\caption{\label{tbl:coef_var} Average CV
for $100$, $300$ and $1000$ answers}
\begin{tabular}{l|l | l | l |}

System	&	\multicolumn{1}{c|}{Mondial}	&	\multicolumn{1}{c|}{Partial DBLP}	&	\multicolumn{1}{c|}{Full DBLP}	\\ \hline
GTF		&	0.54, 0.65, 0.57	&	0.57, 0.57, 0.47	&	0.61, 0.58, 0.57	\\
BANKS	&	0.49, 0.38, 0.40	&	0.51, 0.43, 0.35	&	0.67, 0.62, 0.52	\\
ParLaw	&	0.60, 0.71, 0.73	&	0.66, 0.65, 0.67	&	0.59, 0.59	\\
SPARK	&	0.23, 0.24, 0.22	&	0.37, 0.36, 0.38	&	0.62, 0.63, 0.75	\\
BLINKS	&	0.74, 0.73, 0.75	&	0.38, 0.38, 0.38	&		\\ \hline

\end{tabular}
\end{table}

Table~\ref{tbl:coef_var} gives the averages of the
coefficients of variation (CV) for assessing the confidence in the
experimental results.  These averages are for each system, dataset
and number of generated answers. Each average is over all query
sizes.  Missing entries indicate out-of-memory exceptions.  Since
the entries of Table~\ref{tbl:coef_var} are less than $1$, it means
that the experimental results have a low variance.  It is not surprising that
sometimes the CV of GTF and ParLaw are the largest, because these
are the only two systems that can generate all the
answers. Interestingly, on the largest dataset, GTF has the best CV,
which is another testament to its scalability.  SPARK has the
best CV on small datasets, because (as mentioned earlier)
its running times are hardly affected by either the query size or the
selectivity of the keywords.

\subsection{The Importance of  Producing all Answers}\label{sec:all}

The efficiency of GTF is a significant achievement especially in light
of the fact that it does not miss answers (i.e.,~it is capable of
generating all of them). BANKS, the second most efficient system
overall, cannot produce all the answers. In this section, we
experimentally show the importance of not missing answers.

BANKS has two variants: Bidirectional Search (BiS)~\cite{vldbKPCSDK05}
and Multiple Iterator Backward Search (MI-BkS)~\cite{icdeBHNCS02}.
Figures~\ref{fig:runtimeMondial100}--\ref{fig:runtimeDblp100} show the
running times of the former, because it is the faster of the two.
MI-BkS is at least an order of magnitude slower than BiS. For example,
it takes MI-BkS an average (over all query sizes) of $2.7$
seconds to produce 100 answers on Mondial.  However, as noted
in~\cite{vldbKPCSDK05}, BiS misses more answers than MI-BkS,
because it produces just a single answer for each root.  MI-BkS is
capable of producing several answers for the same root, but all of
them must have a shortest path from the root to each keyword
(hence, for example, it would miss the answer $A_3$ of 
Figure~\ref{fig:answers}).
To test how many answers are missed by BANKS, we ran GTF on all the queries and measured
the following in the generated results.  First, the percentage of
answers that are rooted at the same node as some previously created
answer.  Second, the percentage of answers that contain at least one
non-shortest path from the root to some node that contains
  a keyword of the query.\footnote{We did that after removing the
  keyword nodes so that answers have the same structure as in BANKS.}
The former and the latter percentages show how many answers are missed
by BiS and MI-BkS, respectively, compared with GTF.  The results are
given in Table~\ref{tbl:banks_misses}.
For example, it shows that on Mondial, among the
first 100 answers of GTF, there are only $5$ different roots. So, BiS
would miss $95$ of those answers. Similarly, BiS is capable of
generating merely $20$ answers among the first $1,000$ produced by
GTF.  Also the Table~\ref{tbl:banks_misses} shows that MI-BkS is more
effective, but still misses answers.  On Mondial, it misses $32$ and $580$ answers among
the first $100$ and $1,000$, respectively, generated by GTF.

Generally, the percentage of missed answers increases with the number
of generated answers and also depends heavily on the dataset.  It is
higher for Mondial than for DBLP, because the former is highly
connected compared with the latter.  Therefore, more answers rooted at
the same node or containing a non-shortest path from the root to a
keyword could be created for Mondial.  Both GTF and BANKS enumerate answers by
increasing height.  Since BANKS misses so many answers, it must
generate other answers having a greater height, which are likely to be
less relevant.
\begin{table}[t]
\centering
\caption{\label{tbl:banks_misses}Percentage of answers missed by BiS (i.e.,~have same root as some previous answer)
and MI-BkS (i.e.,~contain non-shortest paths)}
\begin{tabular}{l|l | c | c | c |}
Algorithm & Dataset 		& 100 ans. & 300 ans. & 1000 ans. \\ \hline
\multirow{3}{*}{BiS}&Mondial			&	95	&	97	&	98	\\
					&Partial DBLP	&	85	&	88	&	90	\\
					&Full DBLP		&	81	&	86	&	90	\\ \hline
\multirow{3}{*}{MI-BkS}& Mondial			&	32	&	43	&	58	\\
					   &Partial DBLP	&	10	&	12	&	16	\\
					   &Full DBLP		&	1.8	&	6.1	&	5.8	\\ \hline

\end{tabular}
\end{table}
\section{Configuration of the Systems and Methodology}\label{sec:config_details}
To facilitate repeatability of our experimental results,
we describe how we configured each of the tested systems.

For BANKS, we tested both of their variants, but the figures show the
running times of the more efficient one, namely, bidirectional search.
We disabled the search of keywords in the metadata, because it is not
available in some of the other systems.  The size of the heap was set
to $20$, which means that BANKS generated $120$ answers in order to output $100$.
For the sake of a fair comparison, GTF was also set to generate $20$ more
answers than it was required to output.
In addition,
BANKS produced the answers without duplicates, with respect to
undirected semantics.  Therefore, we also executed GTF 
until it produced $120$, $320$ or $1,020$ answers without
duplicates.\footnote{Duplicate elimination in GTF 
is done by first
  removing the keyword nodes so that the structure of the answers is
  the same as in BANKS.}  For clarity's sake,
Section~\ref{sec:experiments} omits these details and describes all
the algorithms as generating $100$, $300$ and $1,000$ answers.  In
BANKS, as well as in GTF,
redundant answers (i.e.,~those having a root with a single child)
were discarded, that is, they were not included among the generated answers.
The other systems produced $100$, $300$ or $1,000$ according to their default
settings.

For SPARK, we used the block-pipeline algorithm, because it was their
best on our datasets. The size of candidate networks (CN)
was bounded by 5. The \emph{completeness factor} was set to 2,
because the authors reported that this value enforces the AND
semantics for almost all queries.

In BLINKS, we used a random partitioner.  The minimum size of a
partition was set to 5, because a bigger size caused an out-of-memory
error.  BLINKS stores its indexes in the main memory, therefore, when
running BILNKS, we reserved 20GB (instead of 10GB) for the Java heap.

BANKS and SPARK use PostgreSQL and MySQL, respectively. 
In those two DBMS, the default buffer size is tiny.
Hence, we increased it to 5GB (i.e.,~the parameters \texttt{shared\_buffers} of
PostgreSQL and \texttt{key\_buffer} of MySQL were set to 5GB).
 
For each dataset, we first executed a long query to warm
the system and only then did we run the experiments on that dataset.
To allow systems that use a DBMS to take full advantage of the buffer
size, we made two consecutive runs on each dataset (i.e.,~executing
once the whole sequence of queries on the dataset followed immediately
by a second run).  This was done for \emph{all} the systems.
The reported times are those of the second run.

The running times do not include the parsing stage or
translating answers into a user-readable format. 
For BLINKS, the time needed to construct the indexes is ignored.

\end{document}